\documentclass[12pt,draftclsnofoot,onecolumn]{IEEEtran}

\usepackage{amsfonts, amsmath, amsthm, amssymb}
\usepackage{mathtools, graphicx, epstopdf, tikz, standalone, siunitx}
\usepackage{xcolor}
\usepackage{booktabs}
\usepackage{multirow}

\newtheorem{prop}{Proposition}

\newtheorem{theorem}{Theorem}
\newtheorem{remark}{Remark}

\newcommand{\R}{\mathcal{R}}

\allowdisplaybreaks
\begin{document}
	\title{Rate Splitting with Wireless Edge Caching:\\A System-Level-based Co-design}
	\author{\IEEEauthorblockN{Eleni Demarchou, \IEEEmembership{Student Member, IEEE}, Constantinos Psomas, \IEEEmembership{Senior Member, IEEE}, and Ioannis Krikidis, \IEEEmembership{Fellow, IEEE}}
	\thanks{E. Demarchou, C. Psomas, and I. Krikidis are with the IRIDA Research Centre for Communication Technologies, Department of Electrical and Computer Engineering, University of Cyprus, Cyprus (e-mail: \{edemar01, psomas, krikidis\}@ucy.ac.cy). Preliminary results of this work have been presented at the IEEE International Conference on Communications (ICC) Workshops, Jun., 2020 \cite{rs}.}}
	\maketitle
\begin{abstract}
Rate splitting (RS) and wireless edge caching are essential means for meeting the quality of service requirements of future wireless networks. In this work, we focus on the cross-layer co-design of wireless edge caching schemes with sophisticated physical layer techniques, which facilitate non-orthogonal multiple access and interference mitigation. A flexible caching-aided RS (CRS) technique is proposed that operates in various modes that specify the cache placement at the receivers. We consider two caching policies: the intelligent coded caching (CC), as well as the well-known most popular content (MPC) policy. Both caching policies are integrated within the design parameters of RS in order to serve multiple cache-enabled receivers. The proposed technique is investigated from a system level perspective by taking into account spatial randomness. We consider a single cell network consisting of center and edge receivers and provide a comprehensive analytical framework for the evaluation of the proposed technique in terms of achieved rates. Specifically, we derive the rate achieved at each receiver under minimum rate constraints while incorporating the cache placement characteristics. Numerical results are presented which highlight the flexibility of the proposed technique and show how caching can be exploited in order to further boost the performance of RS. 
\end{abstract}

\vspace{-2mm}
\section{Introduction}
With the rapid evolution of the Internet of Things (IoT) and the development of machine to machine (M2M) networks, cellular subscriptions are proliferated, yielding to an ever increasing data traffic \cite{IoT}, while video demands have already become more prevalent. According to Cisco, by 2022 nearly 80\% of the mobile data traffic will be due to video streaming \cite{CISCO}. As such, the foreseen wireless networks are expected to serve a massive number of subscribers requiring more data at a higher rate and lower latency \cite{6G}. In order to meet the upcoming increase in spectral efficiency requirements, wireless edge caching is deemed as an essential solution to enhance the content delivery \cite{dimakis}. Specifically, by exploiting the popularity of the network requests and the low cost of storage equipment, the most frequently requested data can be pre-fetched and locally stored at the access points {\cite{SDN}}, or even at the end users cache \cite{cost1}. In this way, a lower playback latency can be attained while offloading the backhaul traffic in a cost efficient manner \cite{cost}. Even though the concept of edge caching has been proposed as a solution to high video traffic, it can be applied to any frequently requested data. The main challenge is to address the appropriate caching policy by taking into account the network architecture \cite{challanges}.

A very well-known scheme for cache placement is the most popular content (MPC) policy, where the available cache is filled with the most popular files of the network \cite{mag}. In \cite{MPC}, the authors consider the MPC policy applied at each small base station (BS) showing the impact of the BS density and their storage size in terms of outage probability and delivery rate. The benefits of the MPC policy have been further exploited in terms of cooperative transmissions. In particular, the work in \cite{letr}, considers a cooperative cache-enabled relay system, and proposes a hybrid caching policy that achieves a balance between the signal cooperation gain achieved by the MPC and the largest content diversity (LCD) gain by caching different files at each relay node. Aiming to achieve a high hit probability, the works in \cite{Probabilistic} and \cite{demarchou} suggest to follow a probabilistic caching placement and randomly store files from the file library by taking into account the files' popularity.

While the aforementioned caching policies, such as MPC, require to store the entire files, another approach of cache placement is the coded caching (CC), where partitions of the files are stored instead \cite{bookCC}. {The CC policy has been introduced by M. A. Maddah-Ali and U. Niesen in \cite{MAN}, and it suggests an intelligently designed cache placement dedicated to a set of receivers and takes into consideration the delivery phase.} Specifically, partitions of the files are carefully stored among the receivers, such that the cache placement can be exploited during the transmission for employing linear network coding. In this way, the coded signal serves simultaneously multiple receivers with a reduced transmitted load. In particular, it is shown that this policy is optimal for minimizing the worst-case load \cite{optimality}. {The gains brought by CC have been investigated in different network deployments including device-to-device networks \cite{caired2d} and multiple antennas setups \cite{CaireMISO}, \cite{EliaMISO} where CC increases the degrees of freedom \cite{EliaMISO}. Moreover, besides caching at the receivers' side, CC has been also exploited in networks with shared caches \cite{pelia} and has been applied for cache-enabled transmitters enabling in this way cooperative transmissions during the delivery phase \cite{CCTx}.} Due to its efficiency, it has been proposed as a key technology for meeting the forthcoming demands in multimedia traffic while obtaining high throughput \cite{6GCC}. Nonetheless, the gains of CC from a system level perspective have not yet been reported.

With the available resources straining out, the development of M2M networks overlaid with the cellular network imposes the need for multiple access schemes that manage to simultaneously offer service to several subscribers. Throughout the evolution of communication systems, various orthogonal multiple access (OMA) schemes were applied consisting orthogonality in several dimensions e.g., time division multiple access, while the current fourth generation systems have adopted orthogonal frequency division multiple access. However, requirements for future networks necessitate the employment of multiple access schemes which are non-orthogonal designed  \cite{6G}, in order to enable concurrent access of the channel resources. Recently, both industry and academia have conducted fundamental efforts on the power domain non-orthogonal multiple access (NOMA). This scheme employs superposition coding in order to serve two receivers simultaneously and successive interference cancellation (SIC) is applied at the receiver with the best channel statistics \cite{jsac}. NOMA with an appropriate power allocation has been proved to improve the performance of OMA in terms of the achievable sum rate \cite{Ding}. Furthermore, it has been shown that the NOMA scheme significantly improves the performance of the user with the worst channel conditions, hence further outperforming OMA in terms of fairness \cite{krikidis}. Due to its superiority against OMA, NOMA has been proposed for the 3GPP-LTE (as multi user superposition transmission) \cite{MUST} and is a promising scheme for future wireless communications. Moreover, applications of NOMA in cache-enabled networks have been studied extensively \cite{cNOMA1}-\cite{cNOMA4} with the available cache being exploited for interference cancellation, hence improving the performance of the network \cite{cNOMA3}, \cite{cNOMA4}. 

The NOMA scheme, is a special case of the so-called rate splitting (RS); a multiple access scheme which consists of an extra degree of freedom in the power domain and allows partial interference cancellation at both receivers. The concept of RS was first introduced in \cite{HANKOBA} and even though the scheme is known for a while, recent studies show the benefits and its applications in different multiple antennas setups by optimizing the beamforming design. Specifically, the authors in \cite{spectral} consider a multiple-input single-output (MISO) broadcast channel and show that RS provides a smooth transition between space division multiple access (SDMA) and NOMA and outperforms them in terms of spectral efficiency with a lower computational complexity. Furthermore, the MISO system with bounded channel state information errors at the transmitter is considered in \cite{maxmin}, where the authors present the gains of RS in terms of max-min fairness. This objective is also investigated in \cite{beamforming}, where the authors consider transmit beamforming in multiple multicast groups and present the benefits of RS through a degrees of freedom analysis. A similar analysis is presented in \cite{BrClNOMA} where the authors show that in a MISO system, NOMA never outperforms RS in terms of sum multiplexing gain and max-mix fair multiplexing gain. Additionally, the enhancements of RS in terms of spectral and energy efficient are also shown in \cite{multicast}, where the authors study an RS-assisted non-orthogonal unicast and multicast transmission system. {Moreover, the RS scheme has been investigated in cache-enabled networks. Specifically, the authors in \cite{RSCC1}, exploit the CC policy and spatial multiplexing gains in a MISO setup and show how this interplay can enhance the delivery time and channel state information quality requirements. A similar setup is also investigated in \cite{RSCC2}, where a generalized degrees of freedom analysis is derived by taking into account both centralized and decentralized placements for cache-enabled users employing the CC policy. Although the current literature sheds light on the efficiency of the RS, its performance in cache-enabled networks, from a system-level perspective is missing from the literature.}

Motivated by the above, in this work, we propose a cross-layer caching-aided RS (CRS) technique which integrates the CC placement with the sophisticated physical layer scheme RS. {In particular, under the CRS technique, the RS scheme is employed by the the transmitter in order to simultaneously serve multiple cache-enabled receivers whilst exploiting the properties of cache placement.} The proposed technique is investigated for a basic single cell downlink network consisting of two classes of receivers; center and edge receivers. By taking into account spatial randomness we provide a rigorous mathematical framework to analyze the performance of the receivers employing the CRS technique, in terms of achieved rates. To the best of our knowledge, this is the first system level analysis that considers the co-design of CC and RS. In particular, the main contributions of this paper are summarized as follows:
\begin{itemize}
	\item We present a communication technique that employs RS and operates in four modes, which define the caching placement at each class of receivers. We consider the CC caching policy as well as an entire-file-based placement i.e., MPC. Each mode of operation implements specific communication techniques, based on the receivers' file requests. In particular, a linear network coding may be applied in the case of CC placement, whereas the cache information can be utilized for interference cancellation, in the case of MPC. {As such, while the properties and applications of MPC or other entire-file placements have been extensively presented in the literature \cite{mag}-\cite{demarchou}, by employing the CRS technique, they can be exploited for mitigating interference.} This request-driven approach is flexible and exploits the benefits of each caching placement under the employment of RS. 
	
	\item The performance of our proposed technique is evaluated by following a probabilistic approach and by considering a fixed power allocation for RS. By using tools from stochastic geometry, we provide closed-form expressions for the distributions of the signal-to-interference-plus-noise ratio (SINR) at the receivers. Moreover, by taking into account minimum rate constraints and, based on the channel statistics, analytical expressions for the achieved rates are provided. We show how the caching placements, integrated within the design parameters of RS, boost the performance of the receivers. Finally, we provide an asymptotic case study, in terms of the transmit power, for the rates achieved at the receivers and indicate how critical the choice of the design parameters is.
	
	\item Numerical results are presented, which validate our analysis and illustrate how the main system parameters affect the performance of the network. We show how the proposed communication technique affects the rate of a receiver, depending on the operating mode and the receivers requests. In particular, the impact of the power allocation on the rates achieved at the receivers is discussed while we demonstrate how this is affected by the caching placements. Through our results, we highlight the performance gains brought by the mutual benefit realized by the co-design of the considered caching placements and RS.
\end{itemize}
The rest of the paper is organized as follows. Section \ref{system} presents the cache-enabled network model and our main assumptions. Section \ref{CRS} presents the proposed CRS technique we study and Section \ref{analysis} provides the analytical framework for the achieved rates. Finally, Section \ref{numerical} presents the numerical results and Section \ref{conclusion} concludes our work.

\underline{Notation:} $\mathbb{P}\left[X\right]$ represents the probability of the event $X$ with an expected value $\mathbb{E}\left[X\right]$; $\mathbf{1}(x)$ is an indicator function which gives $1$ if $x$ is true, otherwise gives $0$; $\big|A\big|$ denotes the cardinality of the set $A$ and $\min\{a,b\}$ returns the minimum value between $a$ and $b$; $a \choose b$ denotes the binomial coefficient; $\gamma(\cdot,\cdot)$ and $\Gamma(\cdot)$ denote the lower incomplete and complete Gamma functions, respectively. Finally, we define $\int_{a}^{b^{-}}f(x)\,dx=\lim_{y\to b}\int_{a}^{y}f(x)\,dx$. In Table \ref{table}, we provide the main notation used throughout the paper. 

\begin{table*}[t]
	\centering
	\caption{Summary of Notation}
	\label{table}
	\begin{tabular}{|l|l||l|l|}
		\hline
		\textbf{Notation} & \textbf{Description} & \textbf{Notation} & \textbf{Description} \\ \hline
		$K$ & Number of receivers in each class & $\omega$ & Pre-log factor \\ \hline
		$M$ & Receiver's cache size & $\zeta$ & SINR threshold for the common stream \\ \hline
		$N$ & Number of files cached in partitions & $\Xi(\omega_i)$ & SINR threshold for the private stream \\ \hline
		Index $n$ & $n$-th receiver, $n \in \{c,e\}$ & $u$ & Fraction of common rate allocation \\ \hline
		$s_0$ & Common stream & $\mathcal{R}^0_b$ & Common stream rate for both receivers \\ \hline
		$s_n$ & $n$-th receiver's private stream & $\mathcal{R}^0_{n}$ & Common stream rate of a single receiver \\ \hline
		$p_{x}$ & Power allocated to stream $s_x$ & $\mathcal{R}^{s_0}_n$ & $n$-th receiver's rate for the common stream \\ \hline
		$\beta$ & Power allocation factor for $s_0$ & $\mathcal{R}_n^{p}$ & $n$-th receiver's rate for the private stream \\ \hline
		$\rho$ & Power allocation factor for $s_{n}$ & $\mathcal{R}_n^{pI}$ & $n$-th receiver's rate for the private stream with interf. \\ \hline
		$\eta_n^x$ & $n$-th receiver's SINR, $x\in\{0,p,pI\}$ & $\mathcal{R}_n$ & $n$-th receiver's rate \\ \hline
		$\pi_{\eta}$ & Coverage probability of SINR $\eta$ & $\mathcal{R}_{sum}$ & Sum rate \\ \hline
		$g_{\eta}$ & PDF of SINR $\eta$ & Index IIC & Information based interference cancellation is applied \\ \hline
	\end{tabular}
\end{table*}
\section{System Model}\label{system}
\subsection{Channel Model}
Consider a single cell downlink network, where the transmitter is located at the center of the cell and has a constant transmit power $P$. Let $\mathcal{D}(r)$ denote a disk centered at the transmitter with radius $r$. A set of $2K$ cache-enabled receivers is located within the disk $\mathcal{D}(r_0)$. The receivers are further classified to center and edge receivers with each class consisting of $K$ receivers. The center receivers are randomly distributed within $\mathcal{D}(r_c)$ while the edge receivers are distributed within the annulus formed by the difference of disks $\mathcal{D}(r_e)$ and $\mathcal{D}(r_0)$, where $r_c<r_e<r_0$ \footnote{The considered topology establishes different channel conditions between the served receivers \cite{circles} i.e., the center receivers have better channel statistics with high probability. This is a common approach for employing the conventional NOMA scheme \cite{circles2}--\cite{rings}, while it is shown in \cite{rs} that under the considered topology, the RS scheme adds more flexibility against NOMA.}. {{The network's topology is depicted in Fig. \ref{topology}}}. Furthermore, we consider that all the wireless signals suffer from both large-scale path-loss and small-scale block Rayleigh fading. As such, the channel between the transmitter and the $i$-th receiver, $i \in \{1,\dots,2K\}$, is given by $L_i\triangleq h_i(1+d_i^\alpha)^{-1}$, where $h_i \sim \exp(1)$, $d_i$ denotes the distance of the $i$-th receiver from the transmitter and $\alpha$ is the path loss propagation exponent. Finally, we take into account additive white Gaussian noise with variance $\sigma^2$.
\begin{figure}[t]\centering
	\includegraphics[width=0.5\linewidth]{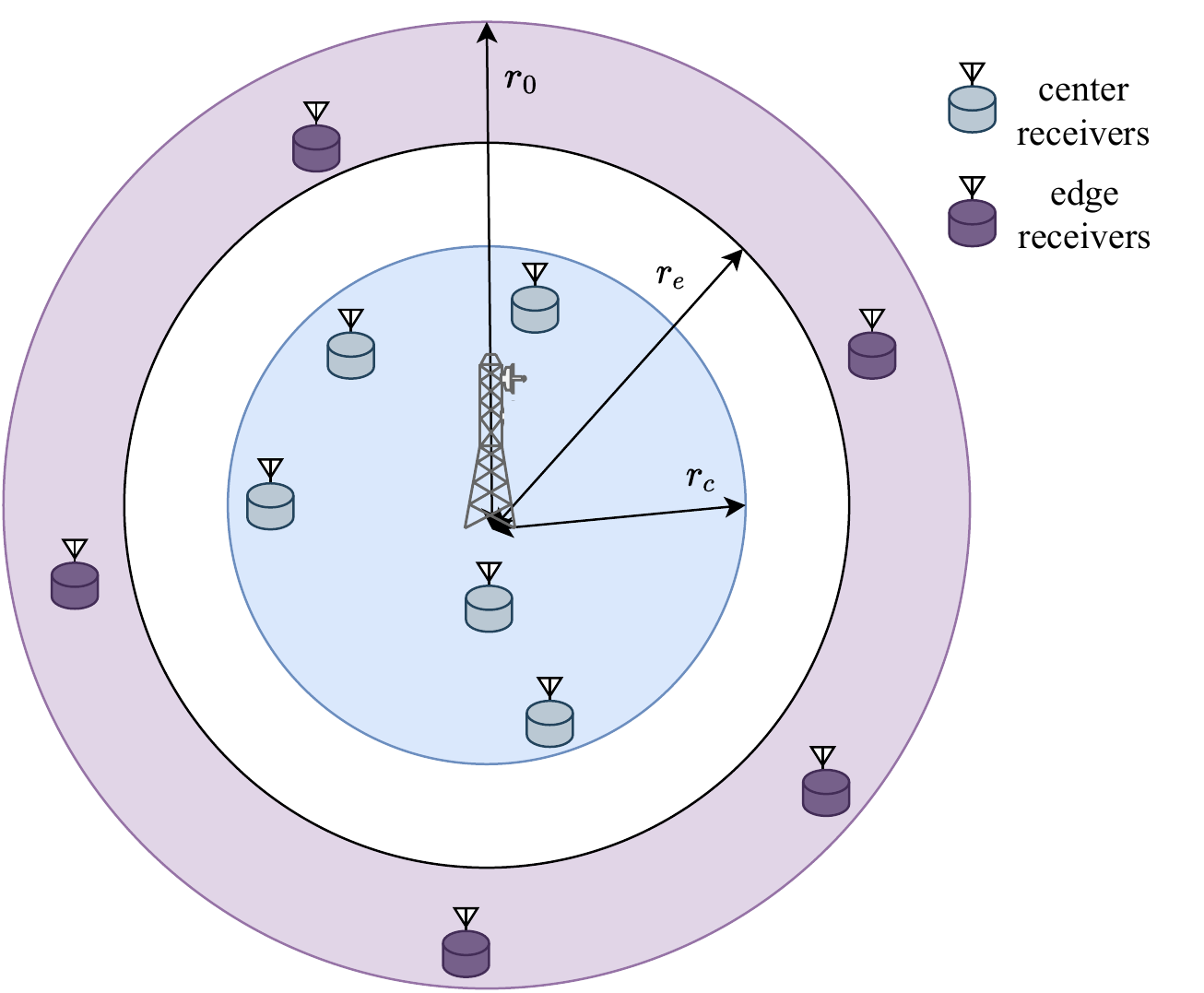}
	\caption{A single cell downlink network with $2K=10$ cache-enabled receivers i.e., $5$ center receivers and $5$ edge receivers.}\label{topology}
\end{figure}
\subsection{Rate Splitting}
The transmitter employs the RS scheme in order to communicate two non-orthogonal signals \cite{spectral}; one for each class of receivers. We focus on the performance of a typical center receiver and a typical edge receiver. As such, throughout this paper the indices $c$ and $e$ refer to the center and edge class, respectively. Accordingly, the transmitter employs the RS scheme to broadcast the messages $m_c$ and $m_e$ to the center and edge receiver, respectively. Each of the two messages is separated in two parts i.e., $m_c=\{m_c^0,m_c^p\}$ and $m_e=\{m_e^0,m_e^p\}$. The parts $m_c^0$ and $m_e^0$ are encoded together and transmitted in a common stream $s_0$ \cite{HANKOBA}. On the other hand, the parts $m_c^p$ and $m_e^p$ are transmitted in two private streams $s_c$ and $s_e$, respectively. Then, the superposition of the three streams is transmitted, by allocating the available power to $s_0$, $s_c$ and $s_e$ with $p_0=\beta P$, $p_c=(1-\beta)\rho P$ and $p_e=(1-\beta)(1-\rho)P$, respectively, where $\beta,\rho\in \left[0,1\right]$. Note that, $\beta$ and $\rho$ are the power allocation factors for the common and private streams, correspondingly. Therefore, the received SINR at the $n$-th receiver for decoding the stream $s_0$ is given by 
\begin{equation}\label{gammaC}
\eta_n^0=\frac{p_0}{p_c+p_e + \sigma^2 L_n^{-1} },
\end{equation} 
where $n \in \{c,e\}$. If the common stream is successfully decoded, then the $n$-th receiver attempts to decode the private stream $s_n$, with SINR $\eta_n^p$ given by 
\begin{equation}\label{gammaP1}
\eta_n^p=\frac{p_n }{p_k+\sigma^2 L_n^{-1}},
\end{equation}
where $k \in \{c,e\}$ and $k \neq n$. On the other hand, if $s_0$ is not successfully decoded, then the $n$-th receiver attempts to decode the private stream with SINR $\eta_n^{pI}$, which is expressed as 
\begin{equation}\label{gammaPI1}
\eta_n^{pI}=\frac{p_n }{p_0+p_k+\sigma^2 L_n^{-1}}.
\end{equation}
It is worth pointing out that, different from NOMA, the RS scheme allows partial interference cancellation at both the receivers, without requiring a channel-based ordering between the receivers.
\subsection{{File Placement - Coded Caching}}
We consider a file library consisting of the $F$ most popular files of the network, each of equal size. The files are sorted in descending order according to their popularity such that the file of rank $f$ is more popular than the file of rank $f+1$. Each receiver has an available storage of $M$ files. {Consider a class of $K$ receivers (center or edge) employing the CC policy introduced in \cite{MAN}, as follows}. Each receiver in the considered class fills its available storage with a fraction $\frac{M}{N}$ of each of the $N$ most popular files, $N > M$. Specifically, each file with a rank $f\leq N$, is partitioned into a set of subfiles, denoted by $\Lambda_f$, where $\big|\Lambda_f\big|= {K \choose t}$ and $t=\frac{MK}{N}$, such that $t \in \{1,\dots, K-1\}$. Let $\mathcal{W}_j$ denote the $j$-th set of receivers representing one of the $\big|\Lambda_f\big|$ ways to form a subset consisting of $t$ receivers, such that  $\big|\mathcal{W}_j\big|=t$ and $j\in\{1,\dots,\big|\Lambda_f\big|\}$. Each subfile is then indexed according to $S_{f,\mathcal{W}_j}$ and the $i$-th receiver\footnote{The receivers are indexed in a random order.}, $i \in \{1,\dots,K\}$ stores the subfiles 
\begin{equation}
\mathcal{T}_{f,i}=\{S_{f,\mathcal{W}_j}\in \Lambda_f ~|~ i \in \mathcal{W}_j,  \forall j \},\,\forall f\leq N,
\end{equation}
where $\big|\mathcal{T}_{f,i}\big|={K-1 \choose t-1}$, \cite{MAN}. Consider for a example a class of $K=5$ receivers, each with a cache size of $M=6$ files storing with CC the $N=10$ most popular files. Each of these files is partitioned into $\big|\Lambda_f\big|=10$ subfiles. In particular, the set of subfiles for the most popular file i.e., $f=1$, is given by $\Lambda_1=\{$$S_{1,\{1,2,3\}},$ $S_{1,\{1,2,4\}},$ $S_{1,\{1,2,5\}},$ $S_{1,\{1,3,4\}},$ $S_{1,\{1,3,5\}},$ $S_{1,\{1,4,5\}},$ $S_{1,\{2,3,4\}},$ $S_{1,\{2,3,5\}},$ $S_{1,\{2,4,5\}},$ $S_{1,\{3,4,5\}}$$\}$. Then, the first receiver i.e., $i=1$, stores the subfiles $\mathcal{T}_{1,1}=\{S_{1,\{1,2,3\}}$, $S_{1,\{1,2,4\}}$, $S_{1,\{1,2,5\}}$, $S_{1,\{1,3,4\}}$, $S_{1,\{1,3,5\}}$, $S_{1,\{1,4,5\}}\}$. In a similar way all the $10$ files are partitioned and the $5$ receivers store files corresponding to their index $i$. As such, each receiver caches $\big|\mathcal{T}\big|=6$ subfiles from each of the $10$ partitioned files. For example, the first receiver caches the subfiles sets $\mathcal{T}_{1,1}$, $\mathcal{T}_{2,1}$, $\mathcal{T}_{3,1}$, $\mathcal{T}_{4,1}$, $\mathcal{T}_{5,1}$, $\mathcal{T}_{6,1}$, $\mathcal{T}_{7,1}$, $\mathcal{T}_{8,1}$, $\mathcal{T}_{9,1}$ and $\mathcal{T}_{10,1}$. This combinatorial file placement enables the possibility of employing a linear network coding at the transmitter, in order to simultaneously serve the $K$ receivers. Detailed specifications for the transmission phase are provided in Section \ref{CRS}.
\section{A Caching-aided RS Technique}\label{CRS}

In this section, we present our proposed CRS technique and its employment in the network. We first provide the file transmission protocol for a class of receivers and then we present the various operating modes of the CRS technique.
\subsection{File transmission Protocol}\label{protocol}
Consider a class of receivers; center or edge, employing the CC policy. The $K$ receivers place requests for files from the library. {The transmitter responds to the $K$ requests with a single signal which conveys information for satisfying either a single request or all $K$ requests. Specifically} 
\begin{itemize}
	\item If all $K$ requests are for files of rank $f\leq N$, then all the requested files have been cached according to CC. In this case, the transmitter can employ a network linear coding and transmit XOR messages such that the non-cached subfiles are delivered to the corresponding receiver while simutaneously serving all the $K$ receivers \cite{MAN}. It is shown that the normalized transmitted load (file size) that occurs for each receiver is $\frac{1-M/N}{1+KM/N}$ \cite{MAN}.	
	\item If at least one request is for a file of rank $f>N$, then XOR transmissions are not feasible and therefore a single, randomly selected receiver is scheduled.
	\begin{itemize}
		\item If the request of the scheduled receiver is of rank $f\leq N$, the request refers to the non-cached partitions of the file i.e., a partial file request (PFR) is placed. As such, the remaining file portion, $1-\frac{M}{N}$, is transmitted.
		\item If the request is of rank $f>N$, then the requested file was not cached at all. Therefore the receiver places an entire file request (EFR), and the complete file is transmitted. 
	\end{itemize}
\end{itemize}
It is worth pointing out that, when the $K$ requests enable the XOR transmissions, not only all the $K$ receivers are served simultaneously, but also each receiver's corresponding load is lower. Furthermore, even when a single receiver is scheduled, due the CC placement policy and the possibility of a PFR, a lower load could occur in comparison to the EFR case\footnote{Note that, besides satisfying either $K$ or a single receiver, more intermediate scheduling options could be taken into account by adjusting accordingly the transmitted load.}, which also corresponds to the case where no files are cached. 

Another case where EFR occurs is when the receivers employ caching policies which require to store the entire files instead of the partitions of files. Evidently, all the non-cached requests forwarded to the transmitter refer to entire files. In this work, besides CC, we also consider the MPC caching policy. Different from the CC policy, when the MPC is employed at a class of receivers, each receiver caches the $M$ most popular files of the library \cite{MPC}. As such, any file request with rank $f\leq M$ can be satisfied locally while the requests of rank $f>M$ are forwarded to the transmitter i.e., EFR. We focus on the performance evaluation of the wireless links, as such when the MPC is employed, we consider that the transmitter schedules a receiver\footnote{We focus on the worst-case scenario where there is at least one request not locally satisfied and therefore is forwarded to the transmitter.} requesting a file with rank $f>M$. That is, unlike CC file transmission capabilities, if MPC is employed at the receivers, only a single request can be satisfied\footnote{Multiple receivers can be satisfied over a single signal in the case of common file requests. This does not affect the performance of the typical receiver in a class since the transmitted load is the same i.e., one for EFR. Hence without loss of generality we assume that each receiver requests for a different file.}. Moreover, we consider that when the $n$-th scheduled receiver employs the MPC policy $n\in\{c,e\}$, then this receiver is able to cancel out interference when the $k$-th, $k \in \{c,e\}$, $k\neq n$ receiver requests for a file of rank $f\leq M$. Whether a partition or the complete file is requested, we consider that the receiver employing MPC can draw out interference by utilizing the information in its local cache. Throughout the rest of the paper, we will refer to this technique as information-based interference cancellation (IIC). Note that, any other caching policy with an entire files placement could be applied since it would still enable IIC. 

\subsection{Operating modes of {{CRS}}}
\begin{table*}[t]
	\centering
	\caption{The main features of the CRS technique}
	\label{summary}
	\begin{tabular}{|c||c|c|c|c|c|}
		\hline
		\textbf{\begin{tabular}[c]{@{}c@{}}Modes\\ of operation\end{tabular}} &
		\textbf{\begin{tabular}[c]{@{}c@{}}Requests\\ center/edge\end{tabular}} &
		\textbf{\begin{tabular}[c]{@{}c@{}}Scheduled \\ center receivers\end{tabular}} &
		\textbf{\begin{tabular}[c]{@{}c@{}}Scheduled \\ edge receivers\end{tabular}} &
		\textbf{$m_c$} &
		\textbf{$m_e$} \\ \hline 
		\textbf{All MPC}                 & $f>M$/$f>M$                  & $1$        & $1$        & EFR & EFR \\ \hline
		\multirow{3}{*}{\textbf{CC/MPC}} & $f\leq M$, $\forall K$/$f>M$    & $K$        & $1$, w/IIC & XOR & EFR \\ \cline{2-6} 
		& $f \leq N$, $\forall K$/$f>M$    & $K$        & $1$        & XOR & EFR \\ \cline{2-6} 
		& $f\leq M$ /$f>M$             & $1$        & $1$, w/IIC & PFR & EFR \\ \hline
		\multirow{3}{*}{\textbf{MPC/CC}} & $f>M$/$f\leq M$, $\forall K$    & $1$, w/IIC & $K$        & EFR & XOR \\ \cline{2-6} 
		& $f>M$/$f\leq N$, $\forall K$    & $1$        & $K$        & EFR & XOR \\ \cline{2-6} 
		& $f>M$/$f\leq M$              & $1$, w/IIC & $1$        & EFR & PFR \\ \hline
		\multirow{9}{*}{\textbf{CC/CC}} &
		\begin{tabular}[c]{@{}c@{}}$f\leq N$, $\forall K$/ $f\leq N$, $\forall K$\end{tabular} &
		$K$ &
		$K$ &
		XOR &
		XOR \\ \cline{2-6} 
		& $f\leq N$, $\forall K$/$f\leq N$ & $K$        & $1$        & XOR & PFR \\ \cline{2-6} 
		& $f\leq N$, $\forall K$/$f>N$    & $K$        & $1$        & XOR & EFR \\ \cline{2-6} 
		& $f\leq N$/$f\leq N$, $\forall K$ & $1$        & $K$        & PFR & XOR \\ \cline{2-6} 
		& $f\leq N$/$f\leq N$          & $1$        & $1$        & PFR & PFR \\ \cline{2-6} 
		& $f\leq N$/$f>N$               & $1$        & $1$        & PFR & EFR \\ \cline{2-6} 
		& $f>N$/$f\leq N$, $\forall K$    & $1$        & $K$        & EFR & XOR \\ \cline{2-6} 
		& $f>N$/$f\leq N$              & $1$        & $1$        & EFR & PFR \\ \cline{2-6} 
		& $f>N$/$f>N$                  & $1$        & $1$        & EFR & EFR \\ \hline
	\end{tabular}
	
\end{table*}
Given the caching policies CC or MPC, the file transmission protocol is employed independently and simultaneously at the two classes of the receivers through the employment of RS at the transmitter. Accordingly, the proposed CRS technique can be implemented in four operating modes which are described as follows.

\begin{itemize}
	\item \textbf{All MPC}: Each receiver in the network caches the files of rank $f\leq M$. The transmitter, in this case, schedules one receiver from each class. We consider that the two receivers request for two distinct files, both of rank $f>M$. In this case, both receivers place an EFR and since their caches consist of the same files i.e., files of rank $f\leq M$, IIC is not feasible. 
	\item \textbf{Hybrid modes}: One of the two sets employs MPC while the other one employs CC; (i) CC/MPC and (ii) MPC/CC, where X/Y denotes that caching policies X and Y are applied in center and edge receivers, respectively. One of the two receiver classes employs MPC while the other one employs CC. The receiver employing the MPC scheme always places an EFR, while applying IIC if possible. That is, when the requests from the receiver class employing the CC policy, are of rank $f\leq M$. Different from the MPC requests, the transmitted load for serving the class employing CC, varies and depends on the requests of all $K$ receivers as explained in Section \ref{protocol}. Note that, in the case where XOR transmissions are employed, the transmitter serves $K+1$ requests, otherwise two receivers are served.
	\item \textbf{All CC}: Both classes of receivers employ CC, where the subfiles are independently partitioned between the two classes. Since all receivers have partitions of files in their cache and not the entire files, IIC is not feasible by any receiver. However, all the possible loads of the transmitter might occur by any class depending on the served receivers requests. Specifically, in each class, either the XOR transmissions are employed, or the scheduled receiver places a PFR or an EFR. When operating in this mode, depending on the requests, the transmitter is able to satisfy either $2K$ receivers (when XOR transmissions are employed for both classes), or $K+1$ (when XOR transmissions are solely employed for a single class) or two receivers (without XOR transmissions). 
\end{itemize}

{The proposed CRS technique is summarized in Table \ref{summary} and is depicted in Fig. \ref{model} which demonstrates an example of the technique operating under the CC/MPC mode.}
   \begin{figure*}[t]\centering 
   	\includegraphics[width=\linewidth]{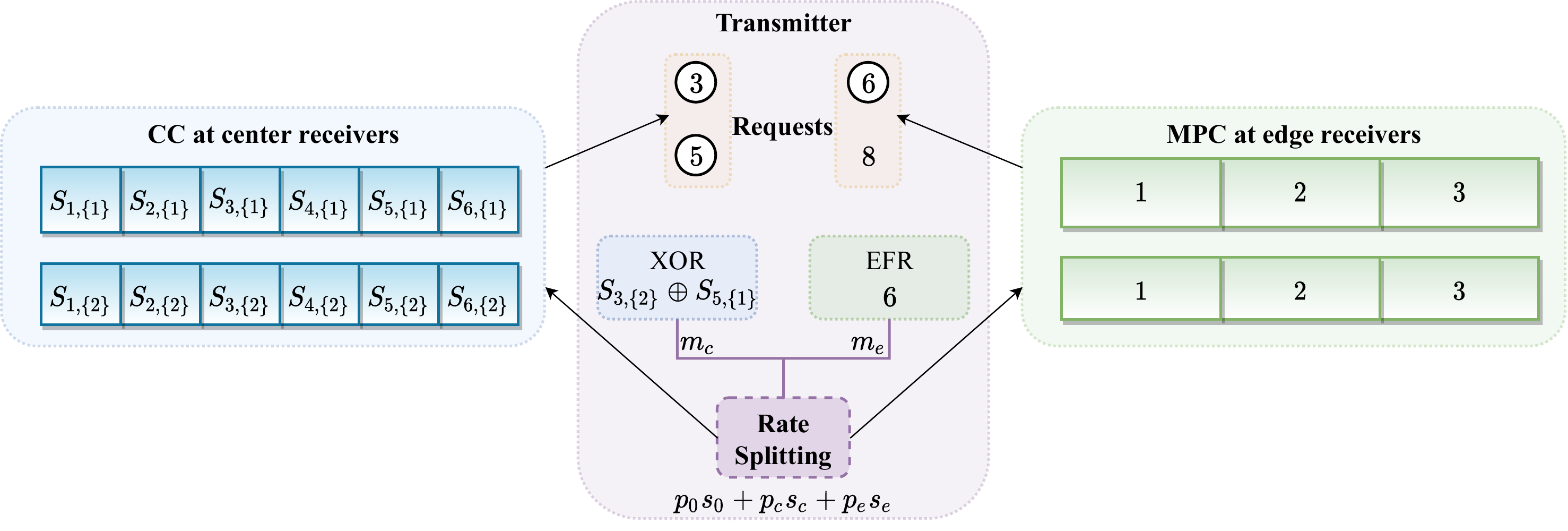}
   	\caption{The proposed CRS operating in CC/MPC mode; $K=2$, $F=10$, $M=3$, $N=6$. The transmitter serves through RS, $K$ center receivers, requesting the non-cached partitions of file ranks $3$ and $5$, with XOR transmissions and a randomly selected edge receiver requesting the entire file of rank $6$. Note that, interference mitigation is only achieved through RS since IIC is not feasible.}\label{model}
   \end{figure*}
\section{CRS Achieved Rates}\label{analysis}
In the following, we derive the rate achieved at the typical center and typical edge receiver when the CRS technique is employed. By taking into account minimum SINR constraints, we derive the rates of the receivers for all the operating modes presented in Section \ref{CRS}. We present how each stream is decoded and analytically demonstrate how the streams' power allocation as well as the various transmission loads affect the achieved rates.  
\subsection{Stream Decoding Rates}
We consider that for decoding each stream, the rate at the $n$-th receiver should exceed a predefined threshold given by $\log_2\left(1+\zeta\right)$ for $s_0$ and by $\log_2\left(1+\xi\right)$ for $s_n$. In order to define the corresponding SINR threshold we need to take into account the transmitted file size. According to the file transmission protocol provided in Section \ref{protocol}, the load of the transmitter depends on the available cache and the served requests. In particular, a file is either entirely transmitted if an EFR is placed, or its non-cached portion is delivered; either to an individual scheduled receiver placing a PFR or in XOR transmissions serving all $K$ receivers in a class of receivers. Thereafter, the amount of the required channel resources also varies. In order to capture this effect, we consider that the rate at a receiver with SINR $\eta$ is expressed by 
\begin{equation}\label{rate}
R(\omega_i,\eta)=\omega_i\log_2\left(1+\eta\right),
\end{equation}
where the pre-log factor $\omega_i$ denotes the $i$-th element of the set $\omega=\Bigl\{ 1,\frac{1}{1-M/N},\frac{1+MK/N}{1-M/N} \Bigr\}$ and corresponds to the inverse of the transmitter's load for the cases EFR, PFR, XOR transmissions, respectively. Note that, even though the XOR transmissions serve $K$ receivers, we focus on their average performance i.e., typical receiver, while the effect of $K$ receivers is integrated in $\omega_3$. Following from the aforementioned when decoding the private stream, the SINR threshold varies according to the requested file size and is given by $\Xi(\omega_i)=\left(1+\xi \right)^{1/\omega_i}-1$. On the other hand, since the stream $s_0$ should be decoded by both the scheduled receivers, the SINR threshold is common and set at $\zeta$. Consider now the case where IIC is not feasible. By taking into account the streams decoding order and the corresponding SINR thresholds, the achieved rate for decoding each stream is defined as follows.
\begin{itemize}
	\item If both receivers successfully decode the common stream $s_0$ i.e., $\eta_c^0>\zeta$ and $\eta_e^0>\zeta$, then $s_0$ has a rate $\R^0_{b}=\mathbb{E}\left[\min\{R(\omega_1,\eta_c^0),R(\omega_1,\eta_e^0)\} \mid \eta_c^0>\zeta, \eta_e^0>\zeta\right]$. As mentioned above the SINR threshold is the same for the two receivers, while the $\min$ operation establishes that the common stream's rate is achievable by both receivers. Let $u \in \left[0,1\right]$ define the common rate allocation factor for the two receivers. Then, $\R^0_{b}$ is allocated to the center and edge receivers with $\omega_i u\R^0_{b}$ and $\omega_j (1-u)\R^0_{b}$, respectively; with $i,j \in \{1,2,3\}$.
	\item If solely the $n$-th receiver successfully decodes $s_0$ i.e.,  $\eta_n^0>\zeta$ and $\eta_k^0 < \zeta$, $k\neq n$, then $\R^0_k=0$ and $\R_n^0=\omega_i\mathbb{E}\left[R(\omega_1,\eta_n^0)\mid \eta_n^0>\zeta, \eta_k^0<\zeta \right]$.
	\item Following the decoding order, if $s_0$ is decoded at the $n$-th receiver ($\eta_n^0>\zeta$), a partial interference cancellation is established and the stream $s_n$ is decoded if $\eta_n^p>\Xi(\omega_i)$, with a rate $\R^p_n(\omega_i)=\mathbb{E}\left[R(\omega_i,\eta_n^p)\mid \eta_n^0>\zeta, \eta_n^p>\Xi(\omega_i) \right]$.
	\item On the other hand, if $s_0$ is not removed i.e., $\eta_n^0<\zeta$, then the stream $s_n$ is successfully decoded if $\eta_n^{pI}>\Xi(\omega_i)$, with a rate $\R^{pI}_n(\omega_i)=\mathbb{E}\left[R(\omega_i,\eta_n^{pI})\mid \eta_n^0<\zeta, \eta_n^{pI}>\Xi(\omega_i) \right]$.
\end{itemize}
In the case where IIC is employed at the $n$-th receiver, then the decoding rates are evaluated as above with the substitution of $\eta_n^0$, $\eta_n^p$ and $\eta_n^{pI}$ with $\eta_{n,\rm{IIC}}^0\triangleq\frac{p_0 }{p_n+\sigma^2 L_n^{-1}}$, $\eta_{n,\rm{IIC}}^{p}\triangleq\frac{p_n L_n}{\sigma^2}$ and $\eta_{n,\rm{IIC}}^{pI}\triangleq\frac{p_n}{p_0+\sigma^2 L_n^{-1}}$, respectively. In the following we derive the achieved rate for decoding each stream as defined above and evaluated by 
\begin{equation}\label{condef}
\R=\mathbb{E}\left[R\big | \epsilon \right]=\frac{\mathbb{E}\left[{\mathbf{1}\{\epsilon\}}R\right]}{\mathbb{P}\left[\epsilon\right]},
\end{equation}
where $\epsilon$ represents the event for which $\R>0$ i.e., if $\mathbb{P}\left[\epsilon\right]=0$ then $\R=0$. For the derivation of the achieved rates the coverage probability $\pi_\eta(t)\triangleq\mathbb{P}\left[\eta>t\right]$ as well as the probability density function (PDF) $g_\eta$ of each SINR are required. We provide the SINR distributions for both the center and the edge receiver in Appendix \ref{app1}. Also, besides $\eta_{n,\rm{IIC}}^p$, the SINR expressions for the $n$-th receiver are upper bounded when $P \to \infty$. As such, throughout the rest of the paper we make use of the following set 
\begin{equation}\label{set}
\vartheta_n=\left\{\frac{p_0}{p_n+p_k}, \frac{p_n}{p_k},\frac{p_n}{p_0+p_k}, \frac{p_0}{p_n} \right\}, 
\end{equation}
where $k \in \{c,e\}$, $k \neq n$ and the $i$-th element of the set is denoted by $\vartheta_{n,i}$. Note that, when the denominator of $\vartheta_{n,i}$ becomes zero, then $\vartheta_{n,i}=\infty$. The expressions $\eta_{n}^0$, $\eta_{n}^p$, $\eta_{n}^{pI}$, $\eta_{n,IIC}^0$ and $\eta_{n,IIC}^{pI}$ are upper bounded by $\vartheta_{n,1}$, $\vartheta_{n,2}$, $\vartheta_{n,3}$, $\vartheta_{n,4}$ and $\vartheta_{n,4}^{-1}$, respectively. Clearly $\vartheta_{c,1}=\vartheta_{e,1}$, hence in the rest of the paper we drop the receiver index and make use of $\vartheta_1$.
\subsection{Achieved Rates: IIC is not feasible}
We now derive the rate achieved at each receiver when IIC is not employed. We first obtain the common rate $\R^0_{b}$ for the case where both received SINRs at the center and edge receivers achieve the minimum threshold $\zeta$. As such, we first express the instantaneous common rate as
\begin{align}
\min\{R(\omega_1,&\eta_c^0),R(\omega_1,\eta_e^0)\mid \eta_c^0>\zeta,\eta_e^0>\zeta\}\nonumber\\
&=\begin{cases}\label{min}
R(\omega_1,\eta_c^0)\mid \eta_c^0>\zeta,\eta_e^0>\zeta,& \eta_c^0<\eta_e^0,\\
R(\omega_1,\eta_e^0)\mid \eta_c^0>\zeta,\eta_e^0>\zeta,&\eta_e^0<\eta_c^0.
\end{cases}
\end{align}
Since $\eta_c^0$ and $\eta_e^0$ are independent then their joint PDF is given by $g_{\eta_c^0}(t)g_{\eta_e^0}(t)$ and by using \eqref{condef}, the rate can be expressed as
\begin{align}
\R^0_{b}=\frac{\mathbb{E}\left[{\mathbf{1}{\{\eta_c^0>\zeta,\eta_e^0>\zeta\}}}\min\{R(\omega_1,\eta_c^0),R(\omega_1,\eta_e^0)\}\right]}{\pi_{\eta_c^0}\left(\zeta\right)\pi_{\eta_e^0}\left(\zeta\right)},
\end{align}
and evaluated by  
\begin{align}\label{commonR}
\R^0_{b}=&\frac{1}{\pi_{\eta_c^0}\left(\zeta\right)\pi_{\eta_e^0}\left(\zeta\right)}\Big(\!\int_{\zeta}^{\vartheta_1^{-}}\!\!\!\int_{\zeta}^{y}\!R(\omega_1,x)g_{\eta_c^0}(x)g_{\eta_e^0}(y)\,dx\,dy+\int_{\zeta}^{\vartheta_1^{-}}\!\!\!\int_{\zeta}^{x}\!R(\omega_1,y)g_{\eta_e^0}(y)g_{\eta_c^0}(x)\,dy\,dx\!\Big),
\end{align}
where $\pi_{\eta_c^0}\left(\zeta\right)$, $\pi_{\eta_e^0}\left(\zeta\right)$, $g_{\eta_c^0}(t)$ and $g_{\eta_e^0}(t)$ are provided in Appendix \ref{app1}. Note that, $\R^0_{b}$ accounts for the case where both receivers successfully decode the common stream $s_0$. When the $n$-th receiver decodes $s_0$ while the $k$-th receiver does not, with $n\neq k$, then due to the independence between the SINRs $\eta_c^0$ and $\eta_e^0$, the rate for decoding $s_0$ at the $n$-th receiver is given by \vspace{-2mm}
\begin{equation}\label{Rcmonostou}
\R_n^0=\frac{1}{\pi_{\eta_n^0}(\zeta)}\int_{\zeta}^{\vartheta_1^{-}}R(\omega_1,t)g_{\eta^0_n}(t)\,dt.
\end{equation}
As such the achieved rate for decoding $s_0$ at the center receiver is 
\begin{equation}
\R_c^{s_0}(\omega_i)=\pi_{\eta_e^0}(\zeta)\omega_i u\R_{b}^0+ \left(1-\pi_{\eta_e^0}(\zeta)\right)\omega_i\R_c^0,
\end{equation}
and, respectively, the achieved rate for decoding $s_0$ at the edge receiver is  
\begin{equation}
\R_e^{s_0}(\omega_j)=\pi_{\eta_c^0}(\zeta)\omega_j (1-u)\R_{b}^0+ \left(1-\pi_{\eta_c^0}(\zeta)\right)\omega_j\R_e^0,
\end{equation}
where $i,j \in \{1,2,3\}$. From the expressions above, it is clear that, with a higher pre-log factor $\omega$, a higher rate occurs and that depends on the served request. Furthermore, the flexibility of RS allows to allocate $\R_{b}^0$ to the receivers, through the factor $u$, in order to enhance the rate at either receiver, for the case where both decode $s_0$. In particular, a higher $u$ provides a higher rate to the center receiver, while with a lower $u$ higher rate is allocated to the edge receiver.

We now focus on the rate for decoding the private stream $s_n$. In the following proposition, we provide the rate $\R_n^p$, which is the rate for decoding $s_n$ when the common stream is successfully decoded. 
\begin{prop}\label{prop3}
	The rate for decoding the private stream $s_n$, $n \in \{c,e\}$, when the common stream $s_{0}$ is successfully decoded is given by 
	\begin{equation}
	\R_n^p(\omega_i)=\frac{1}{\pi_{\eta_n^p}(\Xi(\omega_i))}\int_{\Xi(\omega_i)}^{\vartheta_{n,2}^{-}}R(\omega_i,t)g_{\eta_n^p}(t)\,dt,
	\end{equation}
	if $\frac{\lambda}{\Xi(\omega_i)^{-1}-\vartheta_{n,2}^{-1}}\geq1,  \zeta< \vartheta_1, \Xi(\omega_i)<\vartheta_{n,2}$ and by 
	\begin{equation}
	\R_n^p(\omega_i)=\frac{1}{\pi_{\eta_n^0}\left(\zeta\right)}\int_{\theta_0}^{\vartheta_{n,2}^{-}}R(\omega_i,t)g_{\eta_n^p}(t)\,dt,
	\end{equation}
	if $\frac{\lambda}{\Xi(\omega_i)^{-1}-\vartheta_{n,2}^{-1}}<1,  \zeta< \vartheta_1, \Xi(\omega_i)<\vartheta_{n,2}$; otherwise $\R_n^p=0$, where $\lambda=\vartheta_{n,4}\left(\zeta^{-1}-\vartheta_1^{-1}\right)$ and $\theta_0=\left(\lambda+\vartheta_{n,2}^{-1}\right)^{-1}$.
\end{prop}

\begin{proof}
	See Appendix \ref{proof3}.
\end{proof}
On the other hand, when $s_0$ is not successfully decoded, then the stream $s_n$ is decoded with rate $\R_n^{pI}$, which is provided below.
\begin{prop}\label{prop4}
	The rate for decoding the stream $s_n$, $n \in \{c,e\}$, when $s_0$ is not successfully decoded is given by
	\begin{equation}
	\R_n^{pI}(\omega_i)=\frac{1}{\pi_{\eta_n^{pI}}(\Xi(\omega_i))}\int_{\Xi(\omega_i)}^{\vartheta_{n,3}^{-}}R(\omega_i,t)g_{\eta_n^{pI}}(t)\,dt,
	\end{equation}
	if $\zeta> \vartheta_1$, $\Xi(\omega_i)<\vartheta_{n,3}$ and by 
	\begin{equation}
	\R_n^{pI}(\omega_i)=\frac{1}{\pi_{\eta_n^{pI}}(\Xi(\omega_i))-\pi_{\eta_n^0}\left(\zeta\right)}\int_{\Xi(\omega_i)}^{\theta_{I}}R(\omega_i,t)g_{\eta_n^{pI}}(t)\,dt,
	\end{equation}
	if $\zeta< \vartheta_1, \frac{\lambda}{\Xi(\omega_i)^{-1}-\vartheta_{n,3}^{-1}}<1$, $\Xi(\omega_i)<\theta_{I}$; otherwise $\R_n^{pI}=0$, where $\lambda=\vartheta_{n,4}\left(\zeta^{-1}-\vartheta_1^{-1}\right)$ and $\theta_I=\left(\lambda+\vartheta_{n,3}^{-1}\right)^{-1}$. 
\end{prop}

\begin{proof}
	See Appendix \ref{proof4}.
\end{proof}
Provided with the individual rates for each stream we can now evaluate the rate achieved at each receiver, which is given in the following theorem.
\begin{theorem}\label{theorem1}
	The rate achieved at the $n$-th, $n \in \{c,e\}$ receiver for decoding the streams $s_0$ and/or $s_n$ is given by 
	\begin{equation}
	\R_n(\omega_i)=\R_n^{s_0}(\omega_i)+\frac{\pi_{\eta_n^p}(\Xi(\omega_i))}{\pi_{\eta_n^0}(\zeta)}\R_n^p(\omega_i),
	\end{equation}
	if $\zeta< \vartheta_1, \frac{\lambda}{\Xi(\omega_i)^{-1}-\vartheta_{n,3}^{-1}}>1$, $\frac{\lambda}{\Xi(\omega_i)^{-1}-\vartheta_{n,2}^{-1}}\geq1$, by
	\begin{equation}
	\R_n(\omega_i)=\R_n^{s_0}(\omega_i)+\R_n^p(\omega_i),
	\end{equation}
	if $\zeta< \vartheta_1, \frac{\lambda}{\Xi(\omega_i)^{-1}-\vartheta_{n,3}^{-1}}>1$, $\frac{\lambda}{\Xi(\omega_i)^{-1}-\vartheta_{n,2}^{-1}}<1$, by
	\begin{align}
	\R_n(\omega_i)=\frac{\pi_{\eta_n^0}\left(\zeta\right)}{\pi_{\eta_n^{pI}}(\Xi(\omega_i))}&\left(\R_n^{s_0}(\omega_i)+\R_n^p(\omega_i)\right)+\left(1-\frac{\pi_{\eta_n^0}\left(\zeta\right)}{\pi_{\eta_n^{pI}}(\Xi(\omega_i))}\right)\R_n^{pI}(\omega_i),
	\end{align}
	if $\zeta< \vartheta_1, \frac{\lambda}{\Xi(\omega_i)^{-1}-\vartheta_{n,3}^{-1}}<1$ and by 
	\begin{equation}
	\R_n(\omega_i)=\R_n^{pI}(\omega_i),
	\end{equation}
	if $\zeta> \vartheta_1$, $\Xi(\omega_i)<\vartheta_{n,3}$; otherwise $\R_n(\omega_i)=0$, where $\lambda=\vartheta_{n,4}\left(\zeta^{-1}-\vartheta_1^{-1}\right)$.
\end{theorem}

\begin{proof}
	See Appendix \ref{proofth}.
\end{proof}

\subsection{Achieved Rates: IIC is feasible}
We now consider the case where IIC is employed at the $n$-th receiver. We first focus on the case where IIC is employed at the center receiver i.e., $n=c$ and derive the common rate for the case where both receivers decode $s_{0}$, which is evaluated similar to equation \eqref{commonR} and is given as follows 
\begin{align}\label{IICcommonR1}
&\R^{0,{\rm{IIC}}_c}_{{b}}=\frac{1}{\pi_{\eta_{c,\rm{IIC}}^0}\left(\zeta\right)\pi_{\eta_e^0}\left(\zeta\right)}\nonumber\\
&\times\Big(\int_{\zeta}^{\vartheta_1^{-}}\int_{\zeta}^{y}R(\omega_1,x)g_{\eta_{c,\rm{IIC}}^0}(x)g_{\eta_e^0}(y)\,dx\,dy+\int_{\zeta}^{\vartheta_{c,4}^{-}}\int_{\zeta}^{\min\{x,\vartheta_1^{-}\}}R(\omega_1,y)g_{\eta_e^0}(y)g_{\eta_{c,\rm{IIC}}^0}(x)\,dy\,dx\Big),
\end{align}
where $g_{\eta_{c,\rm{IIC}}^0}(t)$ and $g_{\eta^0_{e}}(t)$ are  provided in Appendix \ref{app1}. Note that, different from $\R^0_{b}$, in this case, the upper bounds of $g_{\eta_c^0}(t)$ and $g_{\eta^0_{2,\rm{IIC}}}(t)$ are not equal i.e., $\vartheta_1\leq\vartheta_{n,4}$, which requires the $\min$ operation in the second integral. For the case where the center receiver decodes the stream $s_{0}$, while the edge receiver does not, then the common rate at the center receiver is given by 
\begin{equation}
\R_c^{0,{\rm{IIC}}_c}=\frac{1}{\pi_{\eta_{c,\rm{IIC}}^0}(\zeta)}\int_{\zeta}^{\vartheta_{c,4}^{-}}R(\omega_1,t)g_{\eta^0_{c,\rm{IIC}}}(t)\,dt.
\end{equation}
As such, the rate for decoding the common stream at the center receiver is given by 
\begin{equation}
\R_c^{s_{0},{\rm{IIC}}_c}(\omega_i)=\pi_{\eta_e^0}(\zeta)\omega_i u\R^{0,{\rm{IIC}}_c}_{{b}}+ \left(1-\pi_{\eta_e^0}(\zeta)\right)\omega_i\R_c^{0,{\rm{IIC}}_c}.
\end{equation}
Clearly, the common rate achieved at the edge receiver is indirectly affected by the IIC at the center receiver and is given as follows 
\begin{align}
\R_e^{s_{0},{\rm{IIC}}_c}(\omega_i)=\pi_{\eta_{c,\rm{IIC}}^0}(\zeta)\omega_i &(1-u)\R^{0,{\rm{IIC}}_c}_{{b}}+ \left(1-\pi_{\eta_{c,\rm{IIC}}^0}(\zeta)\right)\omega_i\R_e^0.
\end{align}
Similarly, with IIC employed by the edge receiver, the common rate when both receivers decode $s_{0}$, is evaluated by 
\begin{align}\label{IICcommonR2}
&\R^{0,{\rm{IIC}}_e}_{{b}}=\frac{1}{\pi_{\eta_{c}^0}\left(\zeta\right)\pi_{\eta_{e,\rm{IIC}}^0}\left(\zeta\right)}\nonumber\\
&\times\Big(\int_{\zeta}^{\vartheta_{e,4}^{-}}\int_{\zeta}^{\min\{y,\vartheta_1^{-}\}}R(\omega_1,x)g_{\eta_c^0}(x)g_{\eta^0_{e,\rm{IIC}}}(y)\,dx\,dy+ \int_{\zeta}^{\vartheta_1^{-}}\int_{\zeta}^{x}R(\omega_1,y)g_{\eta^0_{e,\rm{IIC}}}(y)g_{\eta_c^0}(x)\,dy\,dx\Big),
\end{align}
and by 
\begin{equation}
\R_e^{0,{\rm{IIC}_e}}=\frac{1}{\pi_{\eta_{e,\rm{IIC}}^0}(\zeta)}\int_{\zeta}^{\vartheta_{e,4}^{-}}R(\omega_1,t)g_{\eta^0_{e,\rm{IIC}}}(t)\,dt,
\end{equation}
in the case where the edge receiver decodes the common stream while the center receiver does not. As such, the rates for decoding $s_{0}$ are given by 
\begin{equation}\label{com1IIC}
\R_c^{s_{0},{\rm{IIC}_e}}(\omega_i)=\pi_{\eta^0_{e,\rm{IIC}}}(\zeta)\omega_i u\R_{b}^{0,{\rm{IIC}_e}}+ \left(1-\pi_{\eta^0_{e,\rm{IIC}}}(\zeta)\right)\omega_i\R_c^0,
\end{equation}
and by 
\begin{align}
\R_e^{s_{0},{\rm{IIC}_e}}(\omega_i)=\pi_{\eta_c^0}(\zeta)\omega_i &(1-u)\R_{b}^{0,{\rm{IIC}_e}}+ \left(1-\pi_{\eta_c^0}(\zeta)\right)\omega_i\R_e^{0,{\rm{IIC}_e}},
\end{align}
at the center and edge receiver, respectively.

We now proceed to the derivation of the rate achieved when the receiver $n$ employs IIC and decodes the private stream. 
\begin{prop}\label{prop5}
	The rate for decoding the stream $s_n$, $n \in \{c,e\}$, when $s_0$ is successfully decoded and IIC is employed at the $n$-th receiver, is given by 
	\begin{equation}
	\R_n^{p,{\rm{IIC}}_n}(\omega_i)=\frac{1}{\pi_{\eta_{n,\rm{IIC}}^p}(\Xi(\omega_i))}\int_{\Xi(\omega_i)}^{\infty}R(\omega_i,t)g_{\eta_{n,\rm{IIC}}^p}(t)\,dt,
	\end{equation}
	if $\Xi(\omega_i)\left(\vartheta_{n,4}\zeta^{-1}-1\right)\geq1,  \zeta< \vartheta_{n,4}$, and by
	\begin{equation}
	\R_n^{p,{\rm{IIC}}_n}(\omega_i)=\frac{1}{\pi_{\eta_{n,\rm{IIC}}^0}\left(\zeta\right)}\int_{(\vartheta_{n,4}\zeta^{-1}-1)^{-1}}^{\infty}R(\omega_i,t)g_{\eta_{n,{\rm{IIC}}}^p}(t)\,dt,
	\end{equation}
	if $\Xi(\omega_i)\left(\vartheta_{n,4}\zeta^{-1}-1\right)<1,  \zeta< \vartheta_{n,4}$; otherwise $R_n^p=0$.
\end{prop}

\begin{proof}
	The proof is similar to the one provided for Proposition \ref{prop3}, given in Appendix \ref{proof3}.
\end{proof}

\begin{prop}\label{prop6}
	The rate for decoding the stream $s_n$, $n \in \{c,e\}$, when $s_0$ is not successfully decoded is given by 
	\begin{equation}
	\R_n^{pI,{\rm{IIC}}_n}(\omega_i)=\frac{1}{\pi_{\eta_{n,\rm{IIC}}^{pI}}(\Xi(\omega_i))}\int_{\Xi(\omega_i)}^{(1/\vartheta_{n,4})^{-}}\!\!\!\!\!R(\omega_i,t)g_{\eta_{n,\rm{IIC}}^{pI}}(t)\,dt,
	\end{equation}
	if $\zeta> \vartheta_{n,4}$, $\Xi(\omega_i)<\vartheta_{n,4}^{-1}$, and by 
	\begin{align}
	&\R_n^{pI,{\rm{IIC}}_n}(\omega_i)=\frac{1}{\pi_{\eta_{n,\rm{IIC}}^{pI}}(\Xi(\omega_i))-\pi_{\eta_{n,\rm{IIC}}^0}\left(\zeta\right)}\int_{\Xi(\omega_i)}^{\left(\vartheta_{n,4}(\zeta^{-1}+1)-1\right)^{-1}}R(\omega_i,t)g_{\eta_{n,\rm{IIC}}^{pI}}(t)\,dt,
	\end{align}
	if $\zeta< \vartheta_{n,4}, \frac{\vartheta_{n,4} \zeta^{-1}-1}{\left(\Xi(\omega_i)^{-1}-\vartheta_{n,4}\right)}<1$, $\Xi(\omega_i)<\left(\vartheta_{n,4}(\zeta^{-1}+1)-1\right)^{-1}$; otherwise $R_{n}^{pI,\rm{IIC}_n}=0$. 
\end{prop}

\begin{proof}
	The proof is similar to the one provided for Proposition \ref{prop4}, given in Appendix \ref{proof4}.
\end{proof}

Provided with the above, we obtain in the following theorem, the rate achieved at each receiver.

\begin{theorem}
	The rate achieved at the $n$-th receiver employing IIC, for decoding the streams $s_0$ and/or $s_n$, $n \in \{c,e\}$ is given by
	\begin{equation}
	\R_n^{\rm{IIC}_n}(\omega_i)=\R_n^{s_0,{\rm{IIC}}_n}(\omega_i)+\frac{\pi_{\eta_{n,\rm{IIC}}^p}(\Xi(\omega_i))}{\pi_{\eta_{n,\rm{IIC}}^0}(\zeta)}\R_n^{p,{\rm{IIC}}_n}(\omega_i),
	\end{equation}
	if $\zeta< \vartheta_{n,4}, \frac{\vartheta_{n,4} \zeta^{-1}-1}{\Xi(\omega_i)^{-1}-\vartheta_{n,4}}>1$, $\Xi(\omega_i)\left(\vartheta_{n,4}\zeta^{-1}-1\right)\geq1$, by
	\begin{equation}
	\R_n^{\rm{IIC}_n}(\omega_i)=\R_n^{s_0,{\rm{IIC}}_n}(\omega_i)+\R_n^{p,{\rm{IIC}}_n}(\omega_i),
	\end{equation}
	if $\zeta< \vartheta_{n,4}, \frac{\vartheta_{n,4} \zeta^{-1}-1}{\Xi(\omega_i)^{-1}-\vartheta_{n,4}}>1$, $\Xi(\omega_i)\left(\vartheta_{n,4}\zeta^{-1}-1\right)<1$, by 
	\begin{align}
	&\R_n^{\rm{IIC}_n}(\omega_i)=\frac{\pi_{\eta_{n,\rm{IIC}}^0}\left(\zeta\right)}{\pi_{\eta_{n,\rm{IIC}}^{pI}}(\Xi(\omega_i))}\left(\R_n^{s_0,{\rm{IIC}}_n}(\omega_i)+\R_n^{p,\rm{IIC}_n}(\omega_i)\right)+\left(1-\frac{\pi_{\eta_{n,\rm{IIC}}^0}\left(\zeta\right)}{\pi_{\eta_{n,\rm{IIC}}^{pI}}(\Xi(\omega_i))}\right)\R_n^{pI,\rm{IIC}_n}(\omega_i),
	\end{align}
	if $\zeta< \vartheta_{n,4}, \frac{\vartheta_{n,4} \zeta^{-1}-1}{\Xi(\omega_i)^{-1}-\vartheta_{n,4}}<1$, and by 
	\begin{equation}
	\R_n^{\rm{IIC}_n}(\omega_i)=\R_n^{pI,\rm{IIC}_n}(\omega_i),
	\end{equation}
	if $\zeta> \vartheta_{n,4}$; otherwise $\R_n^{\rm{IIC}_n}=0$.
\end{theorem}

\begin{proof}
	The proof is similar to the one given for Theorem \ref{theorem1}, see Appendix \ref{proofth}.
\end{proof}
Finally, for the case where the $k$-th receiver employs IIC, the rate of the $n$-th receiver, denoted by $\R_n^{\rm{IIC}_k}$, $k\neq n$, is evaluated by Theorem \ref{theorem1} with the substitution of $\R_n^{s_0}$ with $\R_n^{s_0,\rm{IIC}_k}$.

\subsection{Sum rate of the typical receivers}
We now focus on the performance for each of the considered modes; (i) All MPC (ii) CC/MPC, (iii) MPC/CC, and (iv) All CC. By operating in each mode, and based on the receivers' requests, the considered communication techniques may be combined, including XOR transmissions and IIC. In the following theorem, we provide the sum rate achieved for each subcase.
\begin{theorem}{\label{sum}}
	When IIC is not employed, the sum rate is given by 
	\begin{equation}
	\!\!\!\R_\text{sum}(\omega_i,\omega_j)=\frac{q_c\R_c(\omega_i)+q_e\R_e(\omega_j)}{q_c+q_e-q_cq_e},\hspace{+5mm}i,j \in \{1,2,3\}\!
	\end{equation} 
	where $\omega_i$ and $\omega_j$ is the pre-log factor corresponding to the request of the center and edge receivers, respectively and 
	\begin{equation}
	q_n=\begin{cases}
	\pi_{\eta_{n}^0}(\zeta), & \zeta< \vartheta_1, \frac{\lambda}{\Xi(\omega_i)^{-1}-\vartheta_{n,3}^{-1}}>1,\\
	\pi_{\eta_n^{pI}}(\Xi(\omega)), & \zeta< \vartheta_1, \frac{\lambda}{\Xi(\omega_i)^{-1}-\vartheta_{n,3}^{-1}}<1 \rm{\,\,or\,\,} \zeta>\vartheta_1.
	\end{cases}
	\end{equation}
	If IIC is employed at the $k$-th receiver, $k\neq n$, the sum rate is given by
	\begin{equation}
	\R^{\rm{IIC}_k}_\text{sum}(\omega_l)=\frac{q_n\R^{\rm{IIC}_k}_n(\omega_l)+q^{\rm{IIC}}_k\R^{\rm{IIC}_k}_k(\omega_1)}{q_n+q^{\rm{IIC}}_k-q_nq^{\rm{IIC}}_k},
	\end{equation}
	where $\frac{1}{\omega_l}$, $l \in \{2,3\}$, denotes the file size of the request of the $n$-th receiver and
	\begin{equation}
	q^{\rm{IIC}}_k=\begin{cases}
	\pi_{\eta_{k,\rm{IIC}}^0}(\zeta), & \zeta< \vartheta_{n,4}, \frac{\vartheta_{n,4} \zeta^{-1}-1}{\Xi(\omega_i)^{-1}-\vartheta_{n,4}}>1,\\
	\pi_{\eta_{n,\rm{IIC}}^{pI}}(\Xi(\omega)), & \zeta< \vartheta_{n,4}, \frac{\vartheta_{n,4} \zeta^{-1}-1}{\Xi(\omega_i)^{-1}-\vartheta_{n,4}}<1 \rm{\,\,or\,\,} \zeta>\vartheta_{n,4},
	\end{cases}
	\end{equation}
	where $\lambda=\vartheta_{n,4}\left(\zeta^{-1}-\vartheta_1^{-1}\right)$.
\end{theorem}

\begin{proof}
	See Appendix \ref{proofsum}.
\end{proof}
Provided with the above, we present now the various operating modes of the proposed communication strategy and the rates that can be achieved when operating in any of the considered modes.
\subsubsection{All MPC}When both classes employ the MPC schemes, both the receivers place an EFR, while IIC is not feasible. As such the sum rate is evaluated by $\R_\text{sum}(\omega_1,\omega_1)$.
\subsubsection{CC/MPC}In this mode, while the edge receiver places an EFR, depending on the request of the scheduled center receiver, the following subcases might occur.
\begin{itemize}
	\item If all the $K$ requests of the center receivers are of rank $f\leq M$, then the sum rate is evaluated by $\R^{\rm{IIC}_e}_\text{sum}(\omega_3)$ i.e., XOR transmissions are employed with IIC performed by the edge receiver.
	\item If a single center receiver is scheduled and places a PFR with rank $f\leq M$, then the sum rate is $\R^{\rm{IIC}_e}_\text{sum}(\omega_2)$. That is, the remaining portion of the partially cached file is transmitted to the center receiver while the edge receiver performs IIC.
	\item If all the $K$ requests are of rank $f\leq N$ while they are not all of rank $f\leq M$, then the sum rate is given by $\R_\text{sum}(\omega_3,\omega_1)$ i.e., XOR transmissions are employed while IIC is not feasible.
	\item If a single center receiver is scheduled and requests for a file or rank $M<f\leq N$, then the sum rate is $\R_\text{sum}(\omega_2,\omega_1)$ i.e., a PFR is placed while IIC is not feasible.
	\item If a single receiver is scheduled requesting for a file of rank $f>N$ i.e., placing an EFR (different from the request of the edge receiver), then the rate is evaluated by $\R_\text{sum}(\omega_1,\omega_1)$, which is similar to the first mode rate. 
\end{itemize}
\begin{figure*}[t]	
	\begin{minipage}{0.48\textwidth}		
		\includegraphics[width=0.9\linewidth]{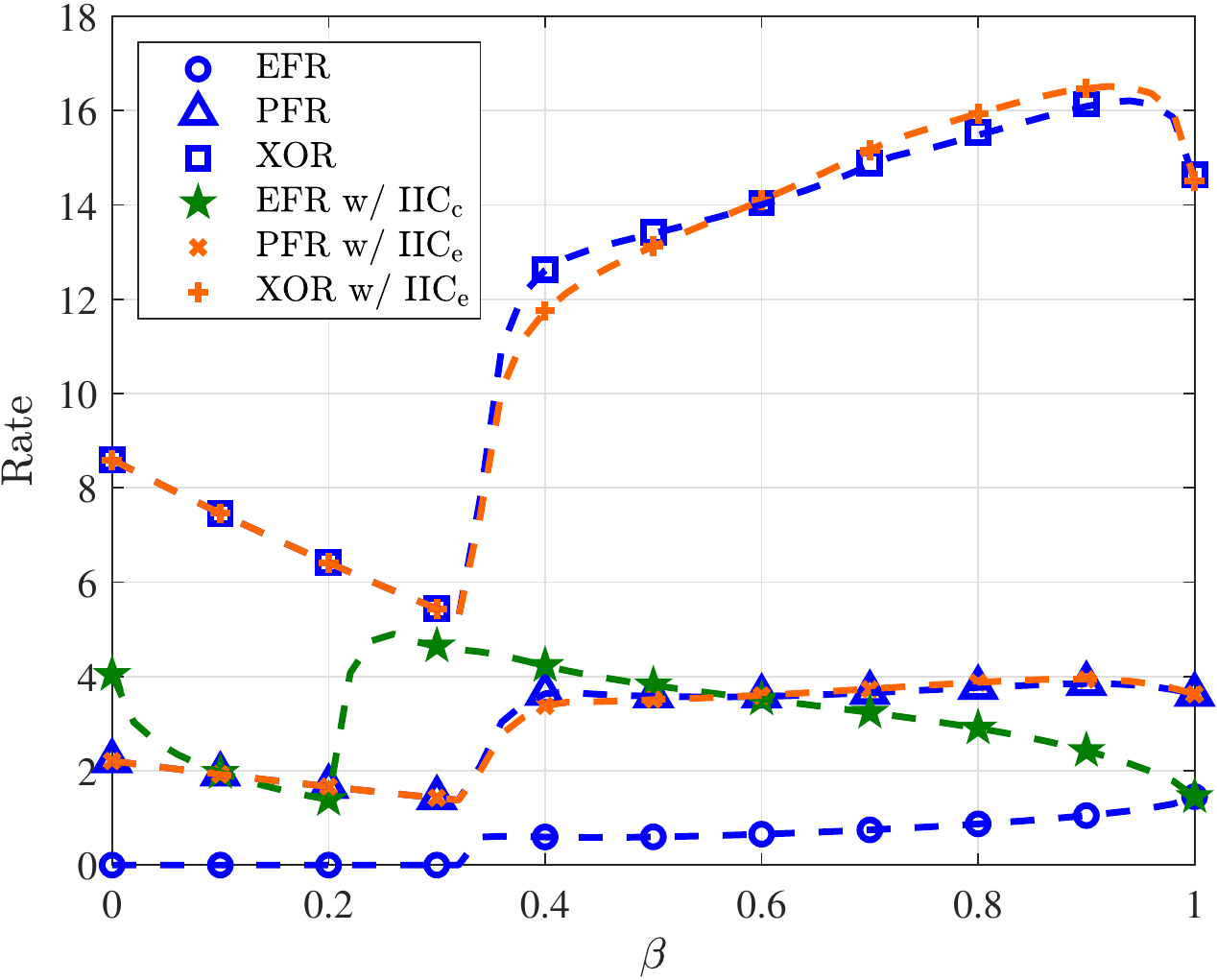}		
		\caption{The rate at the center receiver for all the subcases.}\label{fig2}
	\end{minipage}\hfill
	\begin{minipage}{0.48\textwidth}		
		\includegraphics[width=0.9\linewidth]{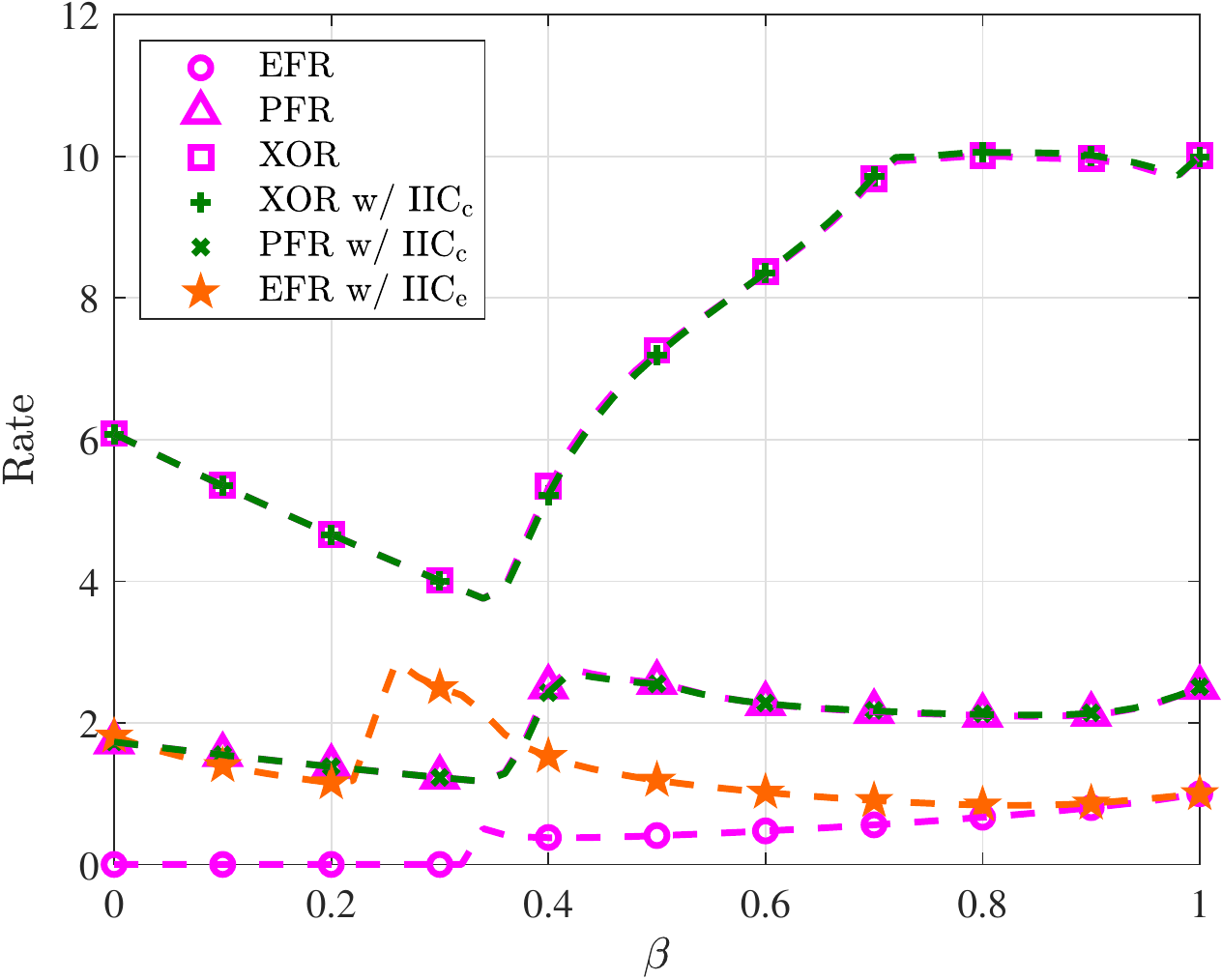}		
		\caption{The rate at the edge receiver for all the subcases.}\label{fig3}
	\end{minipage}\hfill	
\end{figure*}
\subsubsection{MPC/CC}With this mode, the center receiver places an EFR and depending on the request of the edge scheduled receiver several subcases occur. Specifically, the rates that can be achieved are $\R^{\rm{IIC}_c}_\text{sum}(\omega_3)$, $\R^{\rm{IIC}_c}_\text{sum}(\omega_2)$, $\R_\text{sum}(\omega_1,\omega_3)$, $\R_\text{sum}(\omega_1,\omega_2)$ and $\R_\text{sum}(\omega_1,\omega_1)$; each occurring in a similar manner with the previous mode.   
\subsubsection{All CC}In this case where both classes of receivers employ CC, since all receivers have partitions of files in their cache, IIC is not feasible by any receiver. However, all the three possible loads of transmitter might occur by any set. Specifically the sum rate can achieve any of the following depending on the file requests; $\R_\text{sum}(\omega_3,\omega_3)$, $\R_\text{sum}(\omega_3,\omega_2)$, $\R_\text{sum}(\omega_3,\omega_1)$, $\R_\text{sum}(\omega_2,\omega_3)$, $\R_\text{sum}(\omega_2,\omega_2)$, $\R_\text{sum}(\omega_2,\omega_1)$, $\R_\text{sum}(\omega_1,\omega_1)$. 

\subsection{Asymptotic Analysis}
\subsubsection{IIC is not feasible}For the case where no IIC is employed and $P \to \infty$, $\beta \in (0,1)$, $\rho \in (0,1)$ then $\lim_{P \to \infty}\pi_\eta(t)\to1$ while $\eta$ reaches its upper bound. As such, the rates at the receivers become deterministic and are given by 
\begin{align}
\R_c(\omega_i) &= \mathbf{1}\{\vartheta_1\! > \zeta\} (u R(\omega_i,\vartheta_1)\!\nonumber\\
&+ \mathbf{1}\{\vartheta_{c,2}\! > \Xi(\omega_i)\} R(\omega_i,\vartheta_{c,2}))\nonumber\\
&+\mathbf{1}\{\vartheta_1 <\! \zeta\} \mathbf{1}\{\vartheta_{c,3} > \Xi(\omega_i)\} R(\omega_i,\vartheta_{c,3}),
\end{align}

at the center receiver and by 
\begin{align}
\R_e(\omega_i) &= \mathbf{1}\{\vartheta_1 > \zeta\} ((1-u)R(\omega_i\vartheta_1)\nonumber\\
&+\mathbf{1}\{\vartheta_{e,2} > \Xi(\omega_i)\} R(\omega_i,\vartheta_{e,2}))\nonumber\\
&+\mathbf{1}\{\vartheta_1 < \zeta\} \mathbf{1}\{\vartheta_{e,3} > \Xi(\omega_i)\} R(\omega_i,\vartheta_{e,3}),
\end{align} 
at the edge receiver. Subsequently, the asymptotic sum rate is given by $\R_{sum}=\R_c(\omega_i)+\R_e(\omega_j)$, $i,j \in \{1,2,3\}$.

\subsubsection{IIC is feasible}When IIC is employed at the $n$th receiver, then the rate at the receiver might be bounded or not depending on the power allocation and the SINR thresholds. Specifically, $\eta_{n,{\rm{IIC}}}^p$ is not bounded and corresponds to the SINR for decoding the private stream once the common stream has been decoded and extracted. With $P\to\infty$ the coverage probability goes to $1$, as such, even though the common stream's rate is upper bounded ($\eta_{n,{\rm{IIC}}}^0$), since $\eta_{n,{\rm{IIC}}}^p$ is not, the rate increases with $P$. Therefore, the rates with IIC are only bounded when the common stream is not decoded and in this case the achieved rate for the $n$-th receiver is given by 
\begin{equation}
\R_n^{\rm{IIC}_n}(\omega_i)=\mathbf{1}\{\vartheta_{n,4} < \zeta\} \mathbf{1}\{\vartheta_{n,4}^{-1} > \Xi(\omega_i)\} R(\omega_i,\vartheta_{n,4}^{-1}).
\end{equation}
\section{Numerical Results}\label{numerical}
\begin{figure*}[t]	
	\begin{minipage}{0.32\textwidth}
		\includegraphics[width=\linewidth]{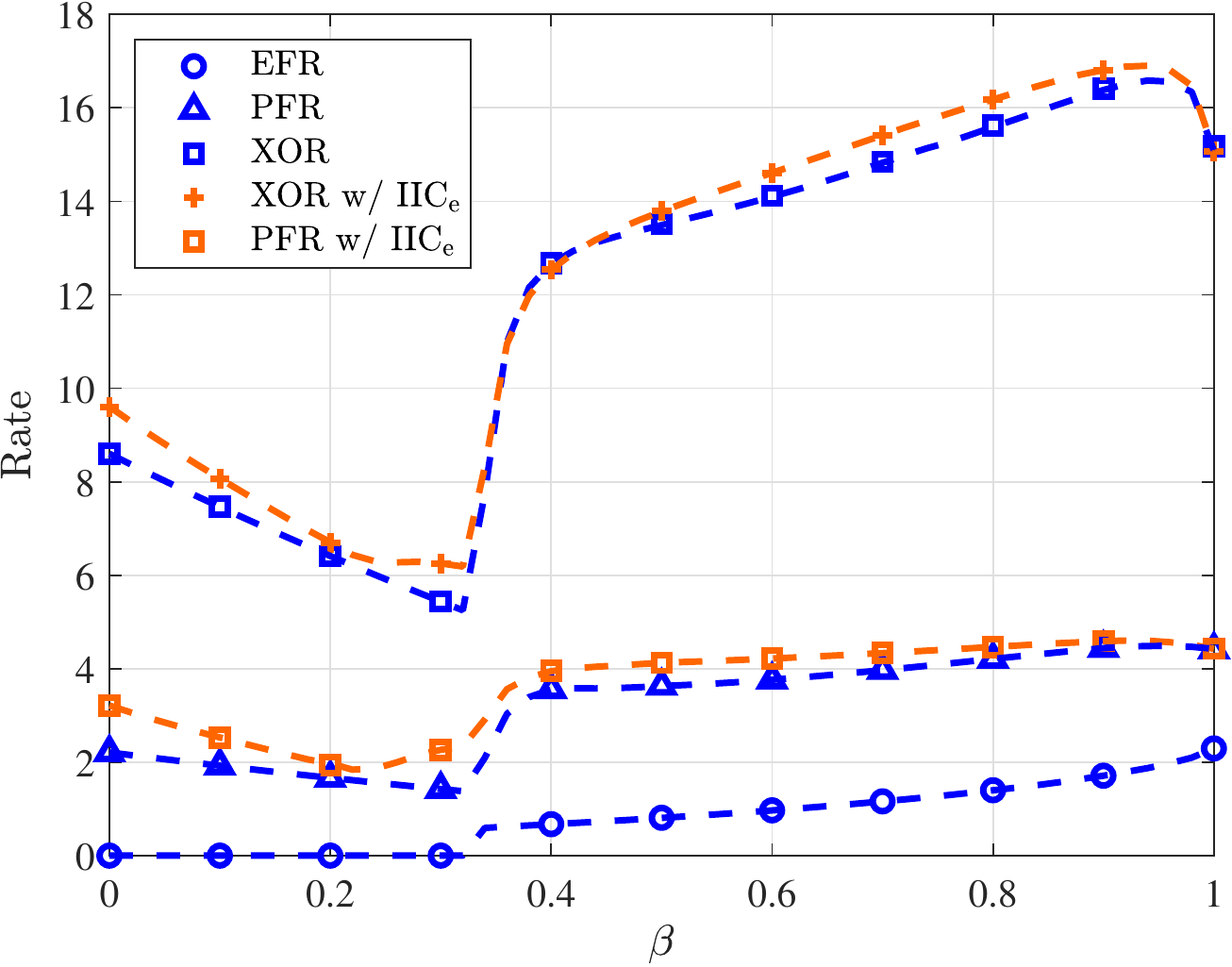}
		
		\caption{The sum rate for all the subcases obtained with the CC/MPC mode.}\label{fig4}
	\end{minipage}\hfill
	\begin{minipage}{0.32\textwidth}
		
		\includegraphics[width=\linewidth]{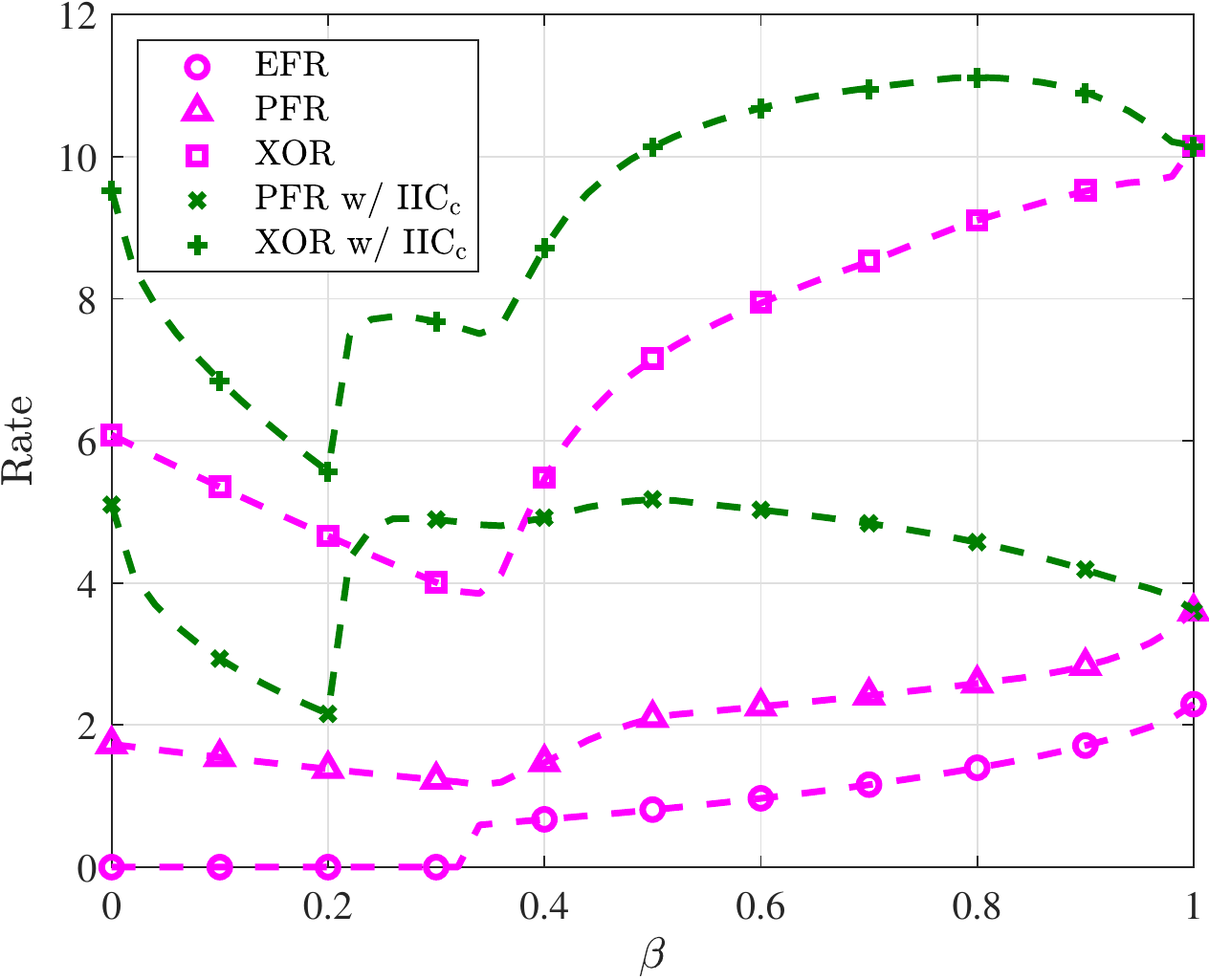}
		
		\caption{The sum rate for all the subcases obtained with the MPC/CC mode.}\label{fig5}
	\end{minipage}\hfill
	\begin{minipage}{0.32\textwidth}
		\includegraphics[width=\linewidth]{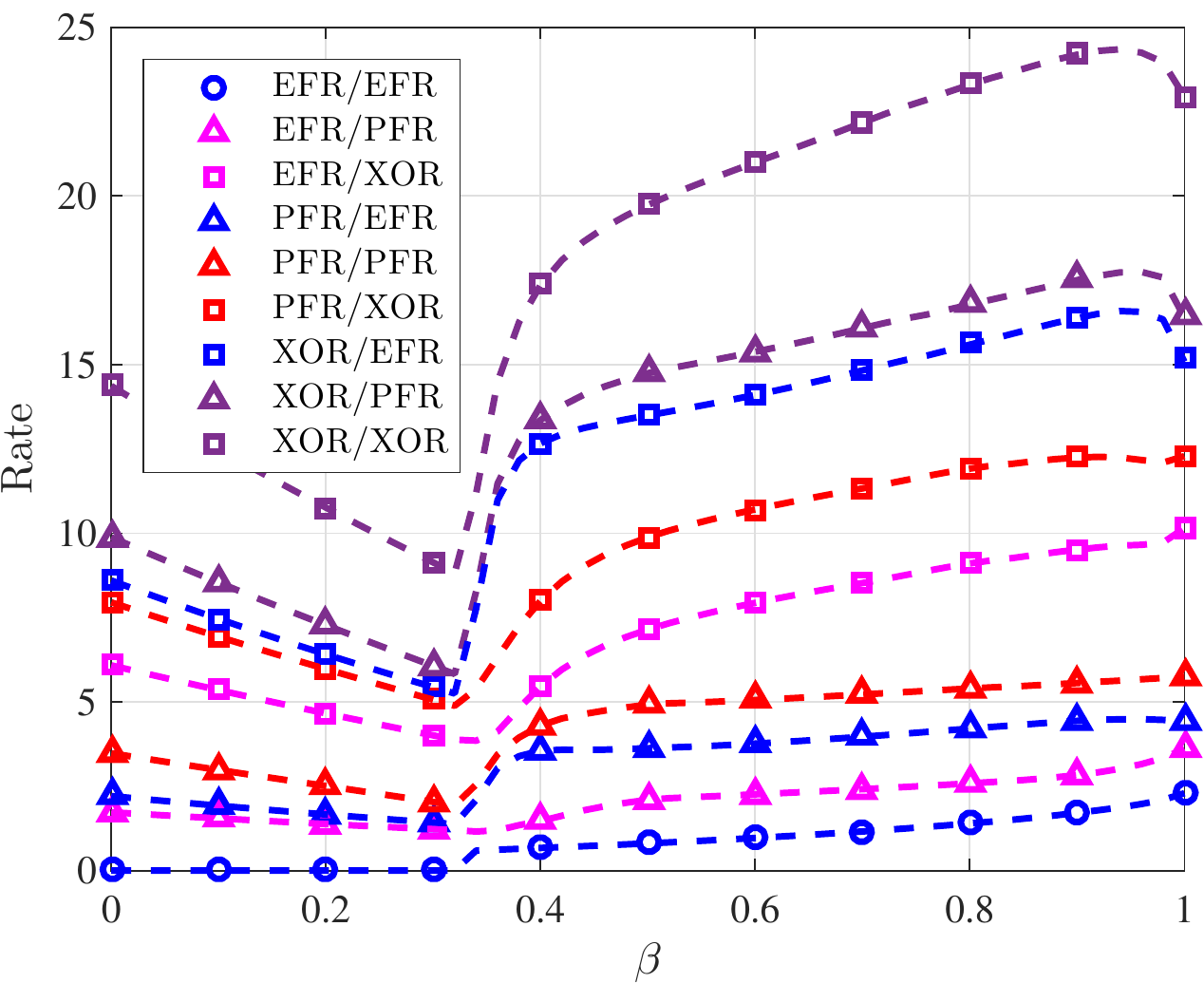}
		
		\caption{The sum rate rate for all the subcases with CC employed at both sets.}\label{fig6}
	\end{minipage}
	
\end{figure*}
In this section, we provide numerical results to evaluate the performance of the typical receivers employing the proposed CRS technique. {Monte-Carlo iterations ($10^5$) were carried out in order to average under different channel realizations (network topology and Rayleigh fading).} Throughout this section markers and dashed lines represent simulation and analytical results, respectively, and unless otherwise stated, we consider the following: $P=10$ W, $\sigma^2=10^{-5}$, $r_c=50$ m, $r_e=60$ m, $r_0=70$ m, $\rho=0.5$, $\zeta=0.5$, $\xi=1$, $u=0.5$, $M=30$, $N=50$ and $K=5$.

In Fig. \ref{fig2} and Fig. \ref{fig3} we plot the rate achieved at the center and edge receiver, respectively, for all the subcases with respect to the fraction of power $\beta$ allocated to the common stream i.e, $p_0=\beta P$. In both figures, it can be observed that, the rate up to a certain value of $\beta$, either decreases or is zero (EFR case). This is due to the upper bounds of the SINR for decoding the common stream, as explained in Appendix \ref{app1}\footnote{In Remarks \ref{remark1} and \ref{remark2}, provided in Appendix \ref{app1}, we indicate the minimum power that should be allocated to a stream such that under minimum rate constraints, the stream is decoded with a non-zero probability.}. At these values, the rate of each receiver is obtained from decoding the private stream with interference from $s_0$, which decreases with $\beta$ since less power is allocated to the private streams $s_c$ and $s_e$. In Fig. \ref{fig2}, when comparing the cases where an EFR is placed (with and without IIC) at the center receiver, we can see that the value of $\beta$ when $s_0$ can be decoded, is lower for the case where IIC is employed, while a non-zero rate is achieved for all the values of $\beta$. On the other hand, when IIC is not employed\footnote{{Note that, the EFR case without IIC can capture the performance achieved when no caching is employed, as it corresponds to requesting a non-cached file without the ability to exploit the local cache.}}, due to the upper bounds of $\eta_{c}^{pI}(\Xi(\omega_1))$, the rate with an EFR is zero until $s_0$ can be decoded. When the stream $s_0$ is decoded, then the private stream $s_c$ can be decoded without interference from $s_0$, if the fraction of power allocated to $s_c$, $\rho$ is higher than a certain value given in Remark \ref{remark1}. Due to that constraint, the rate at the receiver is obtained solely due to decoding $s_0$, and it increases with $\beta$, as expected. Similar observations hold in Fig. \ref{fig3} for the edge receiver placing an EFR (see Remark \ref{remark2}). Furthermore, in both figures, when comparing the cases {EFR, PFR, XOR}, we can see the benefits brought by CC. When XOR transmissions are employed, all the $K$ requests are served with a higher pre-log factor $\omega$ and therefore a higher rate is achieved. Also, we can see that CC is beneficial even when XOR transmissions are not feasible. That is, in the case where a PFR is placed, a better performance is achieved compared to the EFR case, since $\omega_2>\omega_1$. However, in Fig. \ref{fig3}, once the stream $s_0$ is decoded, the rates in the cases of PFR and XOR transmission, increase with $\beta$, as expected, and then they decrease again. This is due to the fact that, less power is allocated to the private stream which can be decoded with higher rate than the common stream. On the other hand, at high values of $\beta$, the rate achieved by decoding the common stream is higher than the one achieved by decoding the private stream, as such the rates at the edge receiver increase again. 

Moreover, we can see that the impact on the $n$-th receiver rate with IIC employed at the $k$-th receiver, is different for the center and edge receivers. In Fig. \ref{fig2} we can see that, for the values of $\beta$ where the receiver is not able to decode $s_0$, the performance of the cases with and without IIC (at the edge receiver), are similar since the SINR constraints (upper bounds) are the same. Once the common stream is decoded, the rates with IIC are lower, while after a certain value of $\beta$, they outperform the rates without IIC. This is due to the fact that, the rate obtained from decoding the common stream depends on $\pi_{\eta^0_{e}}(\zeta)$ and $\pi_{\eta^0_{e,\rm{IIC}}}(\zeta)$, accordingly while $\pi_{\eta^0_{e,\rm{IIC}}}(\zeta)>\pi_{\eta^0_{e}}(\zeta)$. At lower values of $\beta$, where less power is allocated to the common stream, the gains brought by IIC are higher. This however, is a drawback on the center receiver's rate when $\R_c^{0}>u\R_b^{0,\rm{IIC}_e}$ (see eq. \eqref{com1IIC}). On the other hand, at higher values of $\beta$, while $\pi_{\eta^0_{e,\rm{IIC}}}(\zeta)$ has lower gains in comparison to $\pi_{\eta^0_e}(\zeta)$, $\R_b^{0,\rm{IIC}_e}(\zeta)$ outperforms $\R_b^{0}(\zeta)$. At these values, the rates with IIC perform better. Different from the center receivers rate, we can see in Fig. \ref{fig3} that when IIC is employed at the center receiver, due to the difference in the channel conditions i.e., the edge receiver has worst channel conditions with high probability, the impact of IIC employed at the center receiver, is negligible for the edge receiver's rate.

\begin{figure*}[t]	
	\begin{minipage}{0.48\textwidth}
		\includegraphics[width=0.9\linewidth]{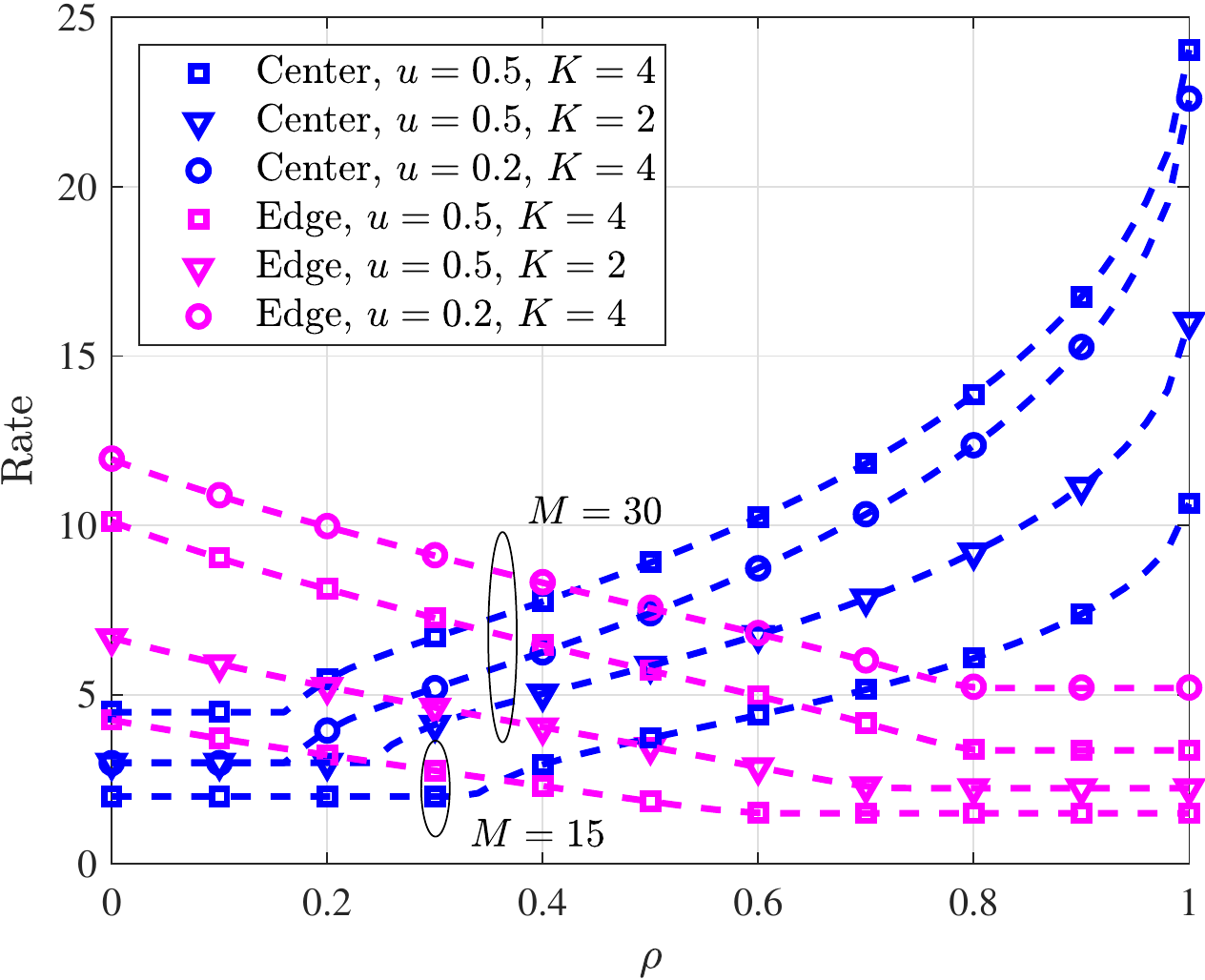}
		\caption{The rate at each receiver with respect to $\rho$, when XOR transmissions are employed without IIC; $N=60$, $\beta=0.7$, $\zeta=0.5$, $\xi=2$.}\label{fig7}
	\end{minipage}\hfill
	\begin{minipage}{0.48\textwidth}
		\includegraphics[width=0.9\linewidth]{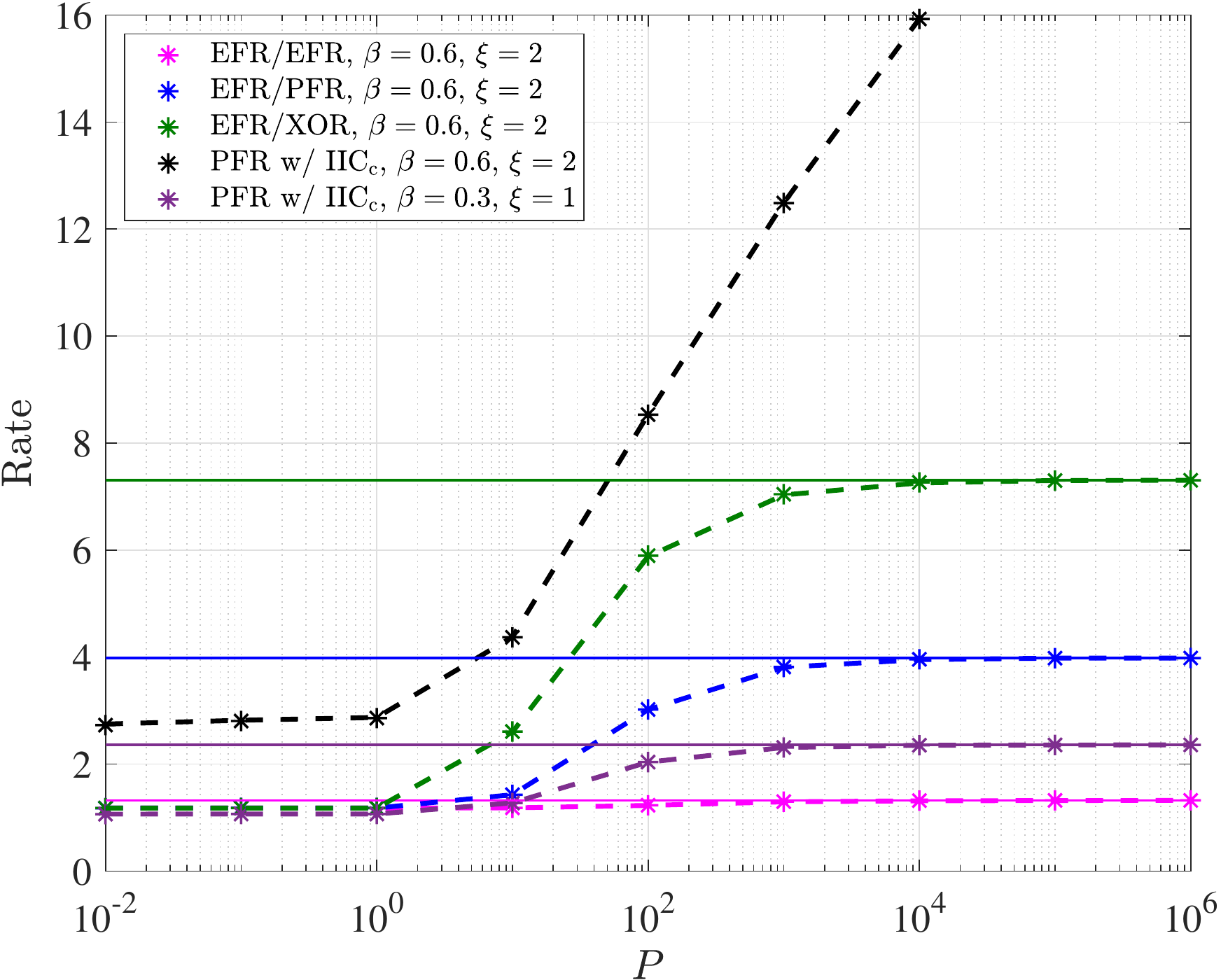}
		\caption{The sum rate with respect to the transmit power $P$; solid lines represent the asymptotic cases; $N=60$, $K=2$, $\beta=0.6$, $\rho=0.5$, $\zeta=1$, $\xi=2$.}\label{fig8}
	\end{minipage}\hfill
\end{figure*}
In Fig. \ref{fig4}, we present the sum rate with respect to $\beta$ for all the subcases obtained when operating in CC/MPC mode. Since MPC is employed at the edge receivers, the scheduled edge receiver always place an EFR i.e., the pre-log factor is $\omega_1$. On the other hand, the pre-log factor for the center receiver varies and depends on all the $K$ requests placed by the center receivers class, as explained in Section \ref{system}. Note that, the rate achieved when EFR is placed by both receivers, is also the only one that can be achieved when MPC is applied at both classes of receivers (all MPC mode). {Furthermore, this case captures the performance achieved when no caching is employed at the receivers, with the transmitter always serving requests for non-cached files.} As expected, this case is outperformed by the rates that can be achieved when the request is cached according to the CC policy; even when the XOR transmissions are not employed i.e., with PFR. Moreover, when IIC is employed at the edge receiver, the gains brought by CC are further boosted, as expected. 

Fig. \ref{fig5}, depicts the sum rate versus $\beta$ for all the subcases that can be achieved when MPC is employed at the center receivers and CC at the edge receivers i.e, MPC/CC mode. In this case, the center receiver always places an EFR, while the load of the transmitter for the edge receiver varies. We can see, as above, the gains achieved by PFR and XOR transmissions which outperform the case where an EFR is placed by both receivers. For the cases of PFR and XOR transmissions, we can see remarkable gains when IIC is employed at the center receiver. The pre-log factor $\omega$ boosts the performance of the edge receiver and combined with IIC at the center receiver, which has better channel conditions, the performance increases significantly. In addition, when IIC is employed, the value of $\beta$ at which the common stream can be decoded, is lower for the center receiver (see Remark \ref{remark1}). As a result, with a PFR by the edge receiver and IIC at the center receiver, there is a range of values for which the sum rate performs better than the case of XOR transmissions at the edge receivers without IIC at the center receiver. 

In Fig. \ref{fig6}, we plot the sum rate with respect to $\beta$ for all the subcases obtained when CC is employed at each class of receivers. In comparison with the other modes, more subcases can be obtained including cases where a pre-log factor higher than $\omega_1$ is obtained at both receivers; for example the case PFR/XOR. Furthermore, we can see that the lowest performance occurs in the case of EFR/EFR, while the highest rate is achieved when XOR transmissions are employed at both sets. Even though, with this mode some of the subcases can be achieved by the previous modes, the IIC is not feasible for further performance boosting since all the receivers partially cache files.

Fig. \ref{fig7} shows the rates achieved at each receiver when XOR transmissions are employed at the center and edge receivers i.e., $\R_c(\omega_3)$ and $\R_e(\omega_3)$ respectively, with respect to $\rho$, for $u\in\{0.2,0.5\}$ and for $K\in\{2,4\}$. It can be seen that, $\R_c(\omega_3)$ increases with $\rho$ while $\R_e(\omega_3)$ decreases. This is expected since $p_c=(1-\beta)\rho$, while $p_e=(1-\beta)(1-\rho)$. For values $\rho\leq\frac{\Xi(\omega_3)}{1+\Xi(\omega_3)}$, the probability $\pi_{\eta_c^p}(\Xi(\omega_3))$ is zero (see Remark \ref{remark1}), while $\pi_{\eta^0_c}(\zeta)>\pi_{\eta_c^{pI}}(\Xi(\omega_3))$. As a result, at those values of $\rho$ the receiver has only rate from decoding the stream $s_0$. When $\rho$ becomes higher that the constraint, the private stream $s_c$ is decoded as well, and the rate increases with $\rho$, as expected. In a similar way, the rate at the edge receiver decreases with $\rho$ since less power is allocated to $s_e$. For values $\rho>\frac{1}{1+\Xi(\omega_3)}$, the probability $\pi_{\eta_e^p}(\Xi(\omega_3))$ becomes zero (see Remark \ref{remark2}) and $\pi_{\eta^0_e}(\zeta)>\pi_{\eta_e^{pI}}(\Xi(\omega_3))$. At these values, the edge receiver's rate occurs solely due to decoding the common stream. Furthermore, we can see that a lower value of $u$ boosts the performance of the edge receiver, while a higher value of $u$ is more beneficial for the center receiver \footnote{{For a specific mode of operation and type of requests, the values of $u$, $\beta$ and $\rho$ can be numerically optimized under a certain objective, by utilizing the provided analytical expressions.}}. This is expected since $u$ is the fraction of the common rate allocated to the center receiver while the fraction $(1-u)$ is allocated to the edge receiver. {Moreover, it is clear that with a lower cache size, a worse performance is achieved at both center and edge receivers. Recall that, the impact of cache size is captured by the pre-log factor i.e., $\omega_3$, which decreases with a lower $M$ and subsequently lower rates can be achieved.} Finally, the performance of both receivers increases with higher $K$, which is expected since with higher $K$, a higher $\omega_3$ is obtained.

In Fig. \ref{fig8}, we plot the sum rate with respect to the transmit power for different loads as well as with IIC employed at the center receiver. When IIC is not employed, the performance of each receiver is bounded due to the SINR upper bound. As such, the rate increases with power until reaching the asymptotic bound, whereas the higher $\omega_i$, the higher the rate, as expected. When IIC is employed, we present two cases (i) $\beta=0.6,\xi=2$ and (ii) $\beta=0.3,\xi=1$. In (i), while the edge receiver's rate is upper bounded, the center receiver's rate is not. Therefore, the rate increases with power and subsequently results in an ever increasing sum rate. This is due to the fact that $\zeta<\vartheta_{c,4}$, while the private stream's SINR is not bounded. On the other hand, in (ii) $\zeta>\vartheta_{c,4}$ and $\Xi(\omega_1)<\vartheta_{c,4}^{-1}$, thus both the center and edge receivers rates are bounded resulting in a bounded sum rate.

\section{Conclusions}\label{conclusion}
In this work, we proposed a CRS technique that employs RS in order to serve cache-enabled receivers by utilizing the benefits brought by the CC and MPC caching placements. The CC gains were captured in terms of pre-log factor while the MPC policy was exploited for interference cancellation. The proposed technique operates in four caching-based modes which can implement several communication techniques according to the receivers requests. By considering spatial randomness we followed a probabilistic approach to provide a complete analytical framework in terms of achieved rate, under minimum rate constraints. We presented numerical results that validated our analysis and we extensively discussed the importance of the power allocation factors and their impact on the achieved rates under each mode of operation. The proposed CRS technique brings further flexibility to the RS through the caching gains and their impact on the achieved rates. Under certain rate requirements caching boosts the receivers' rates and allows more flexibility to the power allocation factors as well as the common rate allocation. {Future extensions of this work include multiple antennas configurations in order to unlock the potentials of the RS scheme and subsequently enhance the performance of the CRS technique.} 

\appendix
\begin{table*}[t]\centering 
	\caption{Parameters for the $n$-th receiver's SINR distributions}
	
	\begin{tabular}{|c||c|c|c|c|c|c|}
		\hline
		$\eta$ & $\eta_n^0$ & $\eta_n^p$ & $\eta_n^{pI}$ & $\eta_{n,IIC}^0$ & $\eta_{n,IIC}^p$ & $\eta_{n,IIC}^{pI}$ \\ \hline
		$\theta$ & $\vartheta_1$ & $\vartheta_{n,2}$ & $\vartheta_{n,3}$ & $\vartheta_{n,4}$ & $\infty$ & $\vartheta_{n,4}^{-1}$ \\ \hline
		$s_\eta$ & $\frac{\sigma^2 t}{p_0-(p_c+p_e)t}$ & $\frac{\sigma^2 t}{p_n-p_k t}$ & $\frac{\sigma^2 t}{p_n-(p_0+p_k)t}$ & $\frac{\sigma^2 t}{p_0-p_n t}$ & $\frac{\sigma^2 t}{p_n}$ & $\frac{\sigma^2 t}{p_n - p_0 t}$ \\ \hline
	\end{tabular}\label{prmtrs}
	
\end{table*}
\subsection{SINR distributions}\label{app1}
In what follows, we provide the coverage probability $\pi_\eta$ and the PDF $g_\eta$ of the SINR $\eta$ for the center and edge receivers. 

\subsubsection{Center receiver}
	The coverage probability for the center receiver with SINR $\eta$, is given by 
	\begin{equation}
	\pi_\eta(t)=\mathbf{1}\left\{t<\theta\right\}\frac{2\exp\left(-s_\eta\right)}{\alpha r_c^2 s_\eta^{2/\alpha}}\gamma \left(\frac{2}{\alpha},s_\eta r_c^\alpha\right),
	\end{equation}
	and the PDF of $\eta$ is given by 
	\begin{align}
	g_\eta(t)=\frac{2\exp\left(-s_\eta\right)}{\alpha t \left(\theta t -1\right)}\left(\frac{1}{\exp\left(s_\eta r_c^\alpha\right)}-\frac{s_\eta+2/\alpha}{s_\eta^{-2/\alpha}r_c^2}\gamma\left(\frac{2}{\alpha}, \frac{s_\eta }{r_c^{-\alpha}}\right)\right),
	\end{align}
	where $\eta$, $\theta$ and $s_\eta$ are given in Table \ref{prmtrs} by setting $n=c$ and $k=e$.

\begin{proof}
Consider the SINR given by \eqref{gammaC} for the center receiver i.e., $n=c$. The coverage probability is evaluated as follows
\begin{align}
\mathbb{P}\left[\eta_c^0>t\right]&=\mathbb{P}\left[\frac{(p_c+p_e) h_c (1+d_c^\alpha)^{-1}+\sigma^2}{p_0 h_c (1+d_c^\alpha)^{-1}}\leq \frac{1}{t}\right]\nonumber\\
&=\mathbb{P}\left[\frac{\sigma^2\left(1+d_c^\alpha\right)}{p_0 h_c}\leq \frac{1}{t}-\frac{p_c+p_e}{p_0}\right].\label{bound}
\end{align}
Based on the upper bound $\vartheta_1$ given in \eqref{set}, we can see that for $t\geq\frac{p_0}{p_c+p_e}$, the coverage probability is zero. Hence, for $t<\vartheta_1$,
\begin{align}
\mathbb{P}\left[\eta_c^0>t\right]&\stackrel{(a)}{=}
\mathbb{E}\left[\exp\left(-\frac{\sigma^2\left(1+x^\alpha\right)}{p_0\left(\frac{1}{t}-\frac{1}{\vartheta_1}\right)}\right)\right]\label{prop1eq1}\nonumber\\
&=\int_{\mathcal{D}(r_c)}\exp\left(-\frac{\sigma^2\left(1+x^\alpha\right)}{p_0\left(\frac{1}{t}-\frac{1}{\vartheta_1}\right)}\right)f_{d_c}(x)\,dx\\
&=\frac{1}{\pi r_c^2}\int_{0}^{2 \pi}\!\!\!\int_{0}^{r_c}\!\!\!x\exp\left(-\frac{\sigma^2\left(1+x^\alpha\right)}{p_0\left(\frac{1}{t}-\frac{1}{\vartheta_1}\right)}\right)\,dx\,d\phi,\label{prop1eq3}
\end{align}
where $(a)$ follows from $h_c \sim \exp(1) $ and $f_{d_c}(x) = 1/\pi r_c^2$ is the PDF of $d_c$. We follow the same procedure for the SINRs $\eta_c^p$ and $\eta_c^{pI}$ and the final general expression for the coverage probability is obtained from \cite[3.326.4]{GRAD}. 

Since the coverage probability $\pi_\eta(t)$ is the complementary cumulative distribution function (CCDF) of SINR, we can use it to derive the PDF of $\eta$ i.e., $g_\eta(t)= d \left(1-\pi_\eta(t)\right)/dt$. Hence, by using \eqref{prop1eq3}, the PDF of $\eta_c^0$ is evaluated as follows
\begin{align}
g_{\eta_c^0}(t)&=\!\frac{-2}{r_c^2}\!\int_{0}^{r_c}x\frac{d}{dt}\left(\exp\!\left(-\frac{\sigma^2\left(1+x^\alpha\right)}{p_0\left(\frac{1}{t}-\frac{1}{\vartheta_1}\right)}\right)\right)\,dx\nonumber\\
&=\!\frac{-2s_{\eta_c^0}}{r_c^2}\!\!\int_{0}^{r_c}\!\frac{\exp\left(-s_{\eta_c^0}\left(1+x^\alpha\right)\right)\left(1+x^\alpha\right)x}{t\left(\vartheta_1 t -1\right)}\,dx,
\end{align}
where the integral is solved with the help of \cite[3.326.4]{GRAD}.
\end{proof}	 

\begin{remark}\label{remark1}
	It is clear that, when $t>\theta$, the center receiver is in outage. Therefore, for a certain $\theta$, this inequality provides the outage conditions with respect to the power allocation factors. Specifically, $\pi_{\eta^0_{c}}(t)=0$ if $\beta\leq \frac{t}{1+t}$, $\pi_{\eta_{c}^p}(t)=0$ if $\rho\leq\frac{t}{1+t}$, $\pi_{\eta_{c}^{pI}}(t)=0$ if  $\{\beta\leq \frac{1}{1+t}, \rho\leq-\frac{t}{\beta t+\beta-t-1}\}$ or $\{\beta>\frac{1}{1+t}, \forall\rho\}$, $\pi_{\eta^0_{c,\rm{IIC}}}(t)=0$ if $\beta\leq \frac{\rho t}{1+ \rho t}$, and $\pi_{\eta_{c,\rm{IIC}}^{pI}}(t)=0$ if $\{\rho=0, \beta>0\}$ or $\{\rho>0, \beta\geq \frac{\rho}{\rho+t}\}$.
\end{remark}

\subsubsection{Edge receiver}
	The coverage probability for the edge receiver with SINR $\eta$, is given by 
	\begin{align}
	\pi_\eta(t)=\mathbf{1}&\left\{t<\theta\right\}\frac{2}{\alpha\left(r_0^2-r_e^2\right) s_{\eta}^{2/\alpha}\exp\left(s_{\eta}\right)}\left(\gamma\left(\frac{2}{\alpha},\frac{s_{\eta}}{r_0^{-\alpha}} \right)-\gamma\left(\frac{2}{\alpha},\frac{s_{\eta}}{r_e^{-\alpha} }\right)\right),
	\end{align}
	and the PDF of $\eta$ is given by
	\begin{align}
	g_\eta(t)=\frac{2\exp \left(-s_\eta\right)}{\alpha t \left(r_e^2-r_0^2\right)\left(\theta t-1\right)}&\Bigg(\frac{r_e^2}{\exp\left(s_\eta r_e^{\alpha}\right)}-\frac{r_0^2}{\exp\left(s_\eta r_0^{\alpha}\right)}\nonumber\\
	&+s_\eta^{-2/\alpha}\left(s_\eta+\frac{2}{\alpha}\right)\left(\gamma\left(\frac{2}{\alpha},\frac{s_\eta }{r_0^{-\alpha}}\right)-\gamma\left(\frac{2}{\alpha},\frac{s_\eta }{r_e^{-\alpha}}\right)\right)\Bigg),
	\end{align}
	where $\eta$, $\theta$ and $s_\eta$ are given in Table \ref{prmtrs} by setting $n=e$ and $k=c$.
\begin{proof}
	The proof is similar to the one provided for the center receiver's distributions with the substitution of $f_{d_c}(x)$ by $f_{d_e}(x)=\frac{1}{\pi\left(r_0^2-r_e^2\right)}$.
\end{proof}

\begin{remark}\label{remark2}
	Similar to Remark \ref{remark1}, we can solve the inequality $t>\theta$, in order to express the effect of the power allocation factors on the coverage probability of the edge receiver. That is, $\pi_{\eta^0_{e}}(t)=0$ if $\beta\leq \frac{t}{1+t}$, $\pi_{\eta_{e}^p}(t)=0$ if $\rho\geq\frac{1}{1+t}$, $\pi_{\eta_{e}^{pI}}(t)=0$ if $\{\beta\leq \frac{1}{1+t}, \rho>\frac{\beta t+\beta-1}{\beta t+\beta-t-1}\}$ or $\{\beta>\frac{1}{1+t}, \forall\rho\}$, $\pi_{\eta^0_{e,\rm{IIC}}}(t)=0$ if $\beta\leq \frac{\rho t-t}{\rho t-t-1}$, and $\pi_{\eta_{e,\rm{IIC}}^{pI}}(t)=0$ if $\{\rho=1, \beta>0\}$ or $\{\rho<1, \beta\geq \frac{\rho-1}{\rho-t-1}\}$.
\end{remark}

\subsection{Proof of Proposition \ref{prop3}}\label{proof3}
The rate achieved for decoding $s_c$, when $s_0$ is successfully decoded, is given by
\begin{equation}
\R_c^p=\frac{\mathbb{E}\left[\mathbf{1}\{\eta_c^0>\zeta, \eta_c^p>\Xi(\omega_i)\} R(\omega_i,\eta_c^p) \right]}{\min\{\pi_{\eta_c^0}(\zeta),\pi_{\eta_c^p}(\Xi(\omega_i))\}},
\end{equation}
since $\R_c^p$ is non-zero only if both the common and the private streams are successfully decoded.
We can deduce from $\eta_c^0>\zeta$ and $\eta_c^p>\Xi(\omega_i)$, that $h_c>\frac{\sigma^2(1+d_c^\alpha)}{p_0 \left(\frac{1}{\zeta}-\frac{p_c+p_e}{p_0}\right)},\text{\hspace{+2mm}}h_c>\frac{\sigma^2(1+d_c^\alpha)}{p_c \left(\frac{1}{\Xi(\omega_i)}-\frac{p_e}{p_c}\right)}$, respectively. When $\frac{\sigma^2(1+d_c^\alpha)}{p_0 \left(\frac{1}{\zeta}-\frac{p_c+p_e}{p_0}\right)} < \frac{\sigma^2(1+d_c^\alpha)}{p_c \left(\frac{1}{\Xi(\omega_i)}-\frac{p_e}{p_c}\right)}$, it implies that $\pi_{\eta_c^0}(\zeta)>\pi_{\eta_c^p}(\Xi(\omega_i))$. Therefore, the rate achieved from decoding $s_c$ goes to its upper limit which is $\vartheta_{c,2}$, and the rate is evaluated as follows 
\begin{align}
\R_c^p &= \frac{\mathbb{E}\left[\mathbf{1}\{\eta_c^p>\Xi(\omega_i)\} \omega_i\log_2(1+\eta_c^p) \right]}{\pi_{\eta_c^p}(\Xi(\omega_i))}\nonumber\\
&=\frac{\omega_i}{\pi_{\eta_c^p}(\Xi(\omega_i))}\int_{\Xi(\omega_i)}^{\vartheta_{c,2}^{-}}\log_2\left(1+t\right)g_{\eta_c^p}(t)\,dt.
\end{align}  
On the other hand, when $\pi_{\eta_c^0}(\zeta)<\pi_{\eta_c^p}(\Xi(\omega_i))$, i.e., $\frac{\sigma^2(1+d_c^\alpha)}{p_0 \left(\frac{1}{\zeta}-\frac{p_c+p_e}{p_0}\right)} > \frac{\sigma^2(1+d_c^\alpha)}{p_c \left(\frac{1}{\Xi(\omega_i)}-\frac{p_e}{p_c}\right)}$, we need to subtract the rate that is not achieved due to the outage of $\pi_{\eta_c^0}(\zeta)$ i.e., $\frac{\sigma^2(1+d_c^\alpha)}{p_c \left(\frac{1}{\Xi(\omega_i)}-\frac{p_e}{p_c}\right)}<h_c<\frac{\sigma^2(1+d_c^\alpha)}{p_0 \left(\frac{1}{\zeta}-\frac{p_c+p_e}{p_0}\right)}$. Therefore, in this case, $\R_c^p$ can be evaluated by
\begin{align}
\R_c^p&=\frac{\mathbb{E}\left[\mathbf{1}\{\eta_c^p>\Xi(\omega_i)\} \omega_i\log_2(1+\eta_c^p)  \right]-\mathbb{E}\left[\mathbf{1}\{\eta_c^0<\zeta, \eta_c^p>\Xi(\omega_i)\} \omega_i\log_2(1+\eta_c^p)  \right]}{\pi_{\eta_c^0}(\zeta)}\nonumber\\
&=\frac{\omega_i}{\pi_{\eta_c^0}(\zeta)}\Big(\int_{\Xi(\omega_i)}^{\vartheta_{c,2}^{-}}\log_2(1+t)g_{\eta_c^p}(t)\,dt-\int_{\Xi(\omega_i)}^{\theta_0}\log_2(1+t)g_{\eta_c^p}(t)\,dt\Big)\nonumber\\
&=\frac{\omega_i}{\pi_{\eta_c^0}(\zeta)}\int_{\theta_0}^{\vartheta_{c,2}^{-}}\log_2(1+t)g_{\eta_c^p}(t)\,dt.
\end{align}
For the evaluation of the upper limit of the rate
$\mathbb{E}\left[\mathbf{1}\{\eta_c^0<\zeta, \eta_c^p>\Xi(\omega_i) \}\omega_i\log_2(1+\eta_c^p)\right]$ i.e., $\theta_0$, we solve the inequality $p_c \left(\frac{1}{\Xi(\omega_i)}-\frac{p_e}{p_c} \right)>	p_0 \left(\frac{1}{\zeta}-\frac{p_c+p_e}{p_0}\right)$,
with respect to $\Xi(\omega_i)$, and is given by $\Xi(\omega_i) < \left(\frac{p_0}{p_c}\left(\zeta^{-1}-\vartheta_1^{-1}\right)+\vartheta_{c,2}^{-1}\right)^{-1}$. For the derivation of $\R_e^p$, we substitute $\eta_c^0$ and $\eta_c^p$ with $\eta_e^0$ and $\eta_e^p$, respectively, and follow similar procedure as above.

\subsection{Proof of Proposition \ref{prop4}}\label{proof4}
The rate for decoding $s_c$, when $s_0$ adds interference to the received SINR, is given by
\begin{align}
\R_c^{pI}=\frac{\mathbb{E}\left[\mathbf{1}\{\eta_c^0<\zeta, \eta_c^{pI}>\Xi(\omega_i)\} \omega_i\log_2(1+\eta_c^{pI}) \right]}{\pi_{\eta_c^{pI}}(\Xi(\omega_i))-\pi_{\eta_c^0}(\zeta)},
\end{align}
since $\R_c^{pI}$ is non-zero with probability $\pi_{\eta_c^{pI}}(\Xi(\omega_i))-\pi_{\eta_c^0}(\zeta)$. If $\zeta\geq\vartheta_1$, then $\pi_{\eta_c^0}(\zeta)=0$ and the rate becomes independent of the condition $\eta_c^0<\zeta$. Therefore $\R_c^{pI}$ is evaluated by
\begin{align}
\R_c^{pI} &= \frac{\mathbb{E}\left[\mathbf{1}\{\eta_c^{pI}>\Xi(\omega_i)\} \omega_i\log_2(1+\eta_c^{pI}) \right]}{\pi_{\eta_c^{pI}}(\Xi(\omega_i))}\nonumber\\
&=\frac{\omega_i}{\pi_{\eta_c^{pI}}(\Xi(\omega_i))}\int_{\Xi(\omega_i)}^{\vartheta_{c,3}^{-}}\log_2(1+t)g_{\eta_2^{pI}}(t)\,dt.
\end{align} 
On the other hand, if $\zeta<\vartheta_1$, then $\R_c^{pI}$ is non-zero if 
$\frac{\sigma^2(1+d_c^\alpha)}{p_c \left(\frac{1}{\Xi(\omega_i)}-\frac{p_0+p_e}{p_c}\right)}<h_c<\frac{\sigma^2(1+d_c^\alpha)}{p_0 \left(\frac{1}{\zeta}-\frac{p_c+p_e}{p_0}\right)}$ i.e., with probability $\pi_{\eta_c^{pI}}(\Xi(\omega_i))-\pi_{\eta_c^0}(\zeta)$. As such, by solving the inequality  
$p_c \left(\frac{1}{\Xi(\omega_i)}-\frac{p_0+p_e}{p_c}\right)>p_0 \left(\frac{1}{\zeta}-\frac{p_c+p_e}{p_0}\right)$
with respect to $\Xi(\omega_i)$, we evaluate the upper limit of $\R_c^{pI}$ which is given by $\Xi(\omega_i)<\left(\frac{p_0}{p_c}\left(\zeta^{-1}-\vartheta_1^{-1}\right)+\vartheta_{c,3}^{-1}\right)^{-1}$.

Finally, when $\pi_{\eta_c^{pI}}(\Xi(\omega_i))<\pi_{\eta_c^0}(\zeta)$ i.e., $\frac{\sigma^2(1+d_c^\alpha)}{p_0 \left(\frac{1}{\zeta}-\frac{p_c+p_e}{p_0}\right)} < \frac{\sigma^2(1+d_c^\alpha)}{p_c \left(\frac{1}{\Xi(\omega_i)}-\frac{p_0+p_e}{p_c}\right)}$, then if $\eta_c^0>\zeta$, $s_0$ is successfully decoded and removed, which implies $s_c$ is decoded with rate $\R_c^p$. If $\eta_c^0$ is in outage then $\eta_c^{pI}$ is also in outage. As a result, $\R_c^{pI}=0$. For the derivation of $\R_e^{pI}$, we substitute $\eta_c^0$ and $\eta_c^{pI}$ with $\eta_e^0$ and $\eta_e^{pI}$, respectively, and follow the above procedure.

\subsection{Proof of Theorem \ref{theorem1}}\label{proofth}
Consider the center receiver, a non-zero rate is achieved if the following streams are decoded; (i) solely $s_{0}$, (ii) $s_{0}$ and $s_c$, (iii) solely $s_c$, where (ii) is the case where partial interference cancellation is successful and hence $s_c$ is decoded without interference from $s_{0}$ while (iii) is the case where $s_c$ is decoded with interference from $s_{0}$. Clearly, all the above cases depend on $\pi_{\eta_c^0}(\zeta)$ and in order to derive the achieved rate we will make use of the rates derived for decoding each stream along with the corresponding conditions as follows. 

When $\zeta<\vartheta_1$, $\Xi(\omega_i)<\vartheta_{c,3}$ and $\pi_{\eta_c^{pI}}(\Xi(\omega_i))<\pi_{\eta_c^0}(\zeta)$, then $\R_c^{pI}=0$; as explained in Appendix \ref{proof4}. As such, the rate is given by either case (i) or (ii) i.e., both conditioned on $\eta_c^0>\zeta$. When $\Xi(\omega_i)<\vartheta_{c,2}$ and $\pi_{\eta_c^0}(\zeta)<\pi_{\eta_c^p}(\Xi(\omega_i))$, then the rate is given by case (ii) and is evaluated by $\R_c=\R_c^{s_{0}}+\R_c^p$ since in that case, $s_c$ is also decoded with probability $1$. On the other hand, when $\pi_{\eta_c^0}(\zeta)>\pi_{\eta_c^p}(\Xi(\omega_i))$ then we have case (i), and with probability $\frac{\pi_{\eta_c^p}(\Xi(\omega_i))}{\pi_{\eta_c^0}(\zeta)}$, the stream $s_c$ is also decoded i.e., case (ii). Therefore the rate is given by $\R_c=\R_c^{s_{0}}+\frac{\pi_{\eta_c^p}(\Xi(\omega_i))}{\pi_{\eta_c^0}(\zeta)}\R_c^p$. Note that, if $\Xi(\omega_i)>\vartheta_{c,2}$, then $\pi_{\eta_c^p}(\Xi(\omega_i))$ and $\R_c^p$ are equal to zero and $\R_c$ has only contribution from $\R_c^{s_{0}}$ i.e., case (i).

When $\zeta<\vartheta_1$, $\Xi(\omega_i)<\vartheta_{c,3}$ and $\pi_{\eta_c^{pI}}(\Xi(\omega_i))>\pi_{\eta_c^0}(\zeta)$, then $\R_c$ is non-zero with probability $\pi_{\eta_c^{pI}}(\Xi(\omega_i))$. Also, due to the fact that $\eta_c^p>\eta_c^{pI}$, then $\pi_{\eta_c^{p}}(\Xi(\omega_i))>\pi_{\eta_c^0}(\zeta)$ also holds. Here, we have the cases (ii) or (iii) i.e., $\R_c=\R_c^{s_{0}}+\R_c^p$, if $s_{0}$ is decoded, otherwise $\R_c=\R_c^{pI}$. Since $\pi_{\eta_c^{pI}}(\Xi(\omega_i))>\pi_{\eta_c^0}(\zeta)$, then $\R_c=\R_c^{pI}$, with probability $\frac{\pi_{\eta_c^{pI}}(\Xi(\omega_i))-\pi_{\eta_c^0}(\zeta)}{\pi_{\eta_c^{pI}}(\Xi(\omega_i))}$ (see Appendix \ref{proof4}); and $\R_c=\R_c^{s_{0}}+\R_c^p$, with probability $1-\frac{\pi_{\eta_c^{pI}}(\Xi(\omega_i))-\pi_{\eta_c^0}(\zeta)}{\pi_{\eta_c^{pI}}(\Xi(\omega_i))}$.

When $\zeta>\vartheta_1$ and $\Xi(\omega_i)<\vartheta_{c,3}$, then $\R_c=\R_c^{pI}$, since $\pi_{\eta_c^0}(\zeta)$ and $\R_c^{s_{0}}$ are equal to zero. Finally, if $\zeta<\vartheta_1$, $\Xi(\omega_i)<\vartheta_{c,3}$, then $\R_c=0$. 

By following similar procedure as above, we derive the rate achieved at the edge receiver. 

\subsection{Proof of Theorem \ref{sum}}\label{proofsum}
We first focus on the case where IIC is not employed. The rate at the $n$-th receiver, $n \in \{c,e\}$, requesting for a size file of $\frac{1}{\omega_i}$, is non-zero with probability $q_n$ which is evaluated by $\pi_{\eta_{n}^0}(\zeta)$ or $\pi_{\eta_{n}^p}(\Xi(\omega_i))$, depending on the streams' thresholds and the power allocation; see Appendix \ref{proofth}. As such, the sum rate is non-zero when either a single or both receivers have non-zero rate i.e., with probability $q_c+q_e-q_cq_e$.

We now consider the case where IIC is employed at the $k$-th receiver. Recall that, IIC is employed when a receiver employs the MPC caching policy, hence the file size of the request is $\frac{1}{\omega_1}$ i.e., EFR. In addition, IIC is employed when the $n$-th receiver requests for file portion i.e., either $\frac{1}{\omega_2}$ or $\frac{1}{\omega_3}$ corresponding to PFR or XOR transmissions. This is due to the fact that IIC is feasible when the requests are of rank $f\leq M$. On the other hand, CC is employed for files $f\leq N$, $N>M$. As such when a receiver employs CC, EFR occurs for requests of rank $f>N$, which are not available in the cache of the receiver employing the MPC scheme.


\begin{thebibliography}{10}
\bibitem{rs}E. Demarchou, C. Psomas, and I. Krikidis, ``Channel statistics-based rate splitting with spatial randomness," in \emph{Proc. IEEE Int. Conf. Commun. Workshops}, Dublin, Ireland (virtual), Jun. 2020, pp. 1--6.
 	
\bibitem{IoT}A. Al-Fuqaha, M. Guizani, M. Mohammadi, M. Aledhari and M. Ayyash, ``Internet of things: A survey on enabling technologies, protocols, and applications," \emph{IEEE Commun. Surveys and Tut.}, vol. 17, no. 4, pp. 2347--2376, Fourth quarter 2015.

\bibitem{CISCO} Cisco, ``Cisco visual networking index: Forecast and trends, 2017--2022 White Paper,'' Feb. 2019. [Online:] http://www.cisco.com

\bibitem{6G}Z. Zhang et al., ``6G wireless networks: Vision, requirements, architecture, and key technologies," \emph{IEEE Veh. Tech. Mag.}, vol. 14, no. 3, pp. 28-41, Sept. 2019.

\bibitem{dimakis} K. Shanmugam, N. Golrezaei, A. G. Dimakis, A. F. Molisch, and G. Caire, ``FemtoCaching: Wireless content delivery through distributed caching helpers," \emph{IEEE Trans. Inf. Theory}, vol. 59, no. 12, pp. 8402--8413, Dec. 2013.

\bibitem{cost1}N. Golrezaei, A. F. Molisch, A. G. Dimakis, and G. Caire, "Femtocaching and device-to-device collaboration: A new architecture for wireless video distribution," \emph{IEEE Commun. Mag.}, vol. 51, no. 4, pp. 142-149, Apr. 2013.

\bibitem{SDN}{Q. Li, W. Shi, X. Ge, and Z. Niu, ``Cooperative edge caching in software-defined hyper-cellular networks," \emph{IEEE J. Sel. Areas Commun.}, vol. 35, no. 11, pp. 2596--2605, Nov. 2017.}

\bibitem{cost}E. Bastug, M. Bennis, and M. Debbah, ``Living on the edge: The role of proactive caching in 5G wireless networks," \emph{IEEE Commun. Mag.}, vol. 52, no. 8, pp. 82--89, Aug. 2014.

\bibitem{challanges}D. Liu, B. Chen, C. Yang, and A. F. Molisch, ``Caching at the wireless edge: Design aspects, challenges, and future directions," \emph{IEEE Commun. Mag.}, vol. 54, no. 9, pp. 22--28, Sep. 2016.

\bibitem{mag}S. He, W. Huang, J. Wang, J. Ren, Y. Huang, and Y. Zhang, ``Cache-enabled coordinated mobile edge network: Opportunities and challenges," \emph{IEEE Wireless Commun. Mag.}, vol. 27, no. 2, pp. 204--211, Apr. 2020.

\bibitem{MPC}E. Bastug, M. Bennis, M. Kountouris, and M. Debbah, ``Cache-enabled small cell networks: Modeling and tradeoffs,” \emph{EURASIP J. Wireless Commun. Netw.}, vol. 2015, no. 1, pp. 1--11, 2015.


\bibitem{letr} G. Zheng, H. A. Suraweera, and I. Krikidis, ``Optimization of hybrid cache placement for collaborative relaying," \emph{IEEE Commun. Lett.,} vol. 21, pp. 442--445, Feb. 2017.

\bibitem{Probabilistic} B. Blaszczyszyn and A. Giovanidis, ``Optimal geographic caching in cellular networks," in \emph{Proc. IEEE Int. Conf. Commun.}, London, UK, Jun. 2015, pp. 3358--3363.

\bibitem{demarchou}E. Demarchou, C. Psomas, and I. Krikidis, ``Caching in large-scale cellular networks with D2D assistance," in \emph{Proc. 24th Int. Conf. Telecommun. (ICT)}, Limassol, Cyprus, May 2017, pp. 1--5.

\bibitem{bookCC}J. Hachem, N. Karamchandani, S. Diggavi, and S. Moharir, ``Coded caching for heterogeneous wireless networks" in \emph{Wireless edge caching: Modeling, analysis, and optimization}, T. X. Vu, E. Baştuğ, S. Chatzinotas, and T. Q. S. Quek, Eds. Cambridge: Cambridge University Press, 2021, ch. 2, pp. 7--36.

\bibitem{MAN}M. A. Maddah-Ali and U. Niesen, ``Fundamental limits of caching," \emph{IEEE Trans. Inf. Theory}, vol. 60, no. 5, pp. 2856--2867, May 2014.

\bibitem{optimality} K. Wan, D. Tuninetti, and P. Piantanida, ``On the optimality of uncoded cache placement," in \emph{Proc. IEEE Inf. Theory Workshop (ITW)}, Cambridge, 2016, pp. 161--165.

\bibitem{caired2d} M. Ji, G. Caire and A. F. Molisch, ``Fundamental limits of caching in wireless D2D networks," \emph{IEEE Trans. Inf. Theory}, vol. 62, no. 2, pp. 849--869, Feb. 2016.

\bibitem{CaireMISO}{S. P. Shariatpanahi, G. Caire, and B. H. Khalaj, ``Physical-layer schemes for wireless coded caching," \emph{IEEE Trans. Inf. Theory}, vol. 65, no. 5, pp. 2792--2807, May 2019.}

\bibitem{EliaMISO}E. Lampiris and P. Elia, ``Adding transmitters dramatically boosts coded-caching gains for finite file sizes," \emph{IEEE J. Sel. Areas Commun.,} vol. 36, no. 6, pp. 1176--1188, Jun. 2018.


\bibitem{pelia} E. Parrinello, A. Unsal and P. Elia, ``Fundamental limits of coded caching with multiple antennas, shared caches and uncoded prefetching," \emph{IEEE Trans. Inf. Theory}, vol. 66, no. 4, pp. 2252--2268, Apr. 2020.

\bibitem{CCTx}{Y. Cao and M. Tao, ``Degrees of Freedom of cache-aided wireless cellular networks," \emph{IEEE Trans. Commun.}, vol. 68, no. 5, pp. 2777--2792, May 2020.}

\bibitem{6GCC}N. Rajatheva \emph{et. al.}, ``Scoring the terabit/s goal: Broadband connectivity in 6G" [Online:] https://arxiv.org/abs/2008.07220

\bibitem{jsac} Z. Ding, X. Lei, G. K. Karagiannidis, R. Schober, J. Yuan and V. K. Bhargava, ``A survey on non-orthogonal multiple access for 5G Networks: Research challenges and future trends," \emph{IEEE J. Sel. Areas Commun.,} vol. 35, pp. 2181--2195, Oct. 2017.

\bibitem{Ding} Z. Ding, P. Fan, and H. V. Poor, ``Impact of user pairing on 5G non orthogonal multiple-access downlink transmissions," \emph{IEEE Trans. Veh. Technol,} vol. 65, no. 8, pp. 6010--6023, Aug. 2016.

\bibitem{krikidis} S. Timotheou and I. Krikidis, ``Fairness for non-orthogonal multiple access in 5G systems," \emph{IEEE Signal Process. Lett.,} vol. 22, pp. 1647--1651, Oct. 2015.

\bibitem{MUST} ``Study on downlink multiuser superposition transmission (MUST) for LTE (Release 13)," 3GPP TR 36.859, Tech. Rep., Dec. 2015.

\bibitem{cNOMA1} Z. Zhao, M. Xu, W. Xie, Z. Ding, and G. K. Karagiannidis, ``Coverage performance of NOMA in wireless caching networks," \emph{IEEE Commun. Lett.}, vol. 22, no. 7, pp. 1458--1461, Jul. 2018.

\bibitem{cNOMA2} K. N. Doan, M. Vaezi, W. Shin, H. V. Poor, H. Shin and T. Q. S. Quek, ``Power allocation in cache-aided NOMA systems: Optimization and deep reinforcement learning approaches," \emph{IEEE Trans. Commun.}, vol. 68, no. 1, pp. 630--644, Jan. 2020.

\bibitem{cNOMA3} L. Xiang, D. W. K. Ng, X. Ge, Z. Ding, V. W. S. Wong and R. Schober, ``Cache-aided non-orthogonal multiple access: The two-user case," \emph{IEEE J. Sel. Topics Signal Process.}, vol. 13, no. 3, pp. 436--451, Jun. 2019.

\bibitem{cNOMA4}X. Song, H. Li, M. Yuan and Y. Huang, ``Coverage performance analysis of wireless caching networks with non-orthogonal multiple access-based multicasting," \emph{IEEE Access}, vol. 7, pp. 164009--164020, 2019.

\bibitem{HANKOBA} T. Han and K. Kobayashi, ``A new achievable rate region for the interference channel," \emph{IEEE Trans. Inf. Theory}, vol. 27, pp. 49--60, Jan. 1981.

\bibitem{spectral} Y. Mao, B. Clerckx, and V. O. Li, ``Rate-splitting multiple access for downlink communication systems: Bridging, generalizing, and outperforming SDMA and NOMA,” \emph{EURASIP J. Wireless Commun. Netw.}, vol. 2018, p. 133, May 2018.


\bibitem{maxmin} H. Joudeh and B. Clerckx, ``Robust transmission in downlink multiuser MISO systems: A rate-splitting approach,” \emph{IEEE Trans. Signal Process.}, vol. 64, pp. 6227--6242, Dec. 2016.

\bibitem{beamforming} H. Joudeh and B. Clerckx, ``Rate-splitting for max-min fair multigroup multicast beamforming in overloaded systems,” \emph{IEEE Trans. Wirel. Commun.}, vol. 16, pp. 7276--7289, Nov. 2017.

\bibitem{BrClNOMA}B. Clerckx \emph{et. al.}, ``Is NOMA efficient in multi-antenna networks? A critical look at next generation multiple access techniques" [Online:] https://arxiv.org/abs/2101.04802

\bibitem{multicast} Y. Mao, B. Clerckx and V. O. K. Li, ``Rate-splitting for multi-antenna non-orthogonal unicast and multicast transmission: Spectral and energy efficiency analysis," \emph{IEEE Trans. Commun.}, vol. 67, no. 12, pp. 8754--8770, Dec. 2019.

\bibitem{RSCC1}{E. Piovano, H. Joudeh, and B. Clerckx, ``On coded caching in the overloaded MISO broadcast channel," in \emph{Proc. IEEE International Symposium on Information Theory (ISIT)}, Aachen, Germany, Jun. 2017, pp. 2795--2799.}

\bibitem{RSCC2}{E. Piovano, H. Joudeh, and B. Clerckx, ``Generalized degrees of freedom of the symmetric cache-aided MISO broadcast channel with partial CSIT," \emph{IEEE Trans. Inf. Theory}, vol. 65, no. 9, pp. 5799--5815, Sept. 2019.}


\bibitem{circles}{A. Sankararaman and F. Baccelli, ``CSMA k-SIC — A class of distributed MAC protocols and their performance evaluation," in \emph{Proc. IEEE Conf. Computer Commun.}, Kowloon, 2015, pp. 2002--2010.}

\bibitem{circles2}Z. Ding, R. Schober, and H. V. Poor, ``A general MIMO framework for NOMA downlink and uplink transmission based on signal alignment," \emph{IEEE Trans. Wireless Commun.}, vol. 15, no. 6, pp. 4438--4454, Jun. 2016.

\bibitem{rings}{C. Zhang, Y. Liu, and Z. Ding, ``Semi-grant-free NOMA: A stochastic geometry model," \emph{IEEE Trans. Wireless Commun.}, Aug. 2021.}

\bibitem{GRAD} I. S. Gradshteyn and I. M. Ryzhik, \emph{Table of Integrals, Series, and Products}. Elsevier, 2007.
\end{thebibliography}
\end{document}